\documentclass[a4paper,unpublished,onecolumn]{quantumarticle}
\pdfoutput=1

\usepackage[colorlinks,allcolors=quantumviolet]{hyperref}
\usepackage[numbers,sort&compress]{natbib}
\usepackage[phfqit,spalign]{tchshortcuts}
\usepackage{amsthm}
\usepackage[
	print-unity-mantissa=false,
	range-phrase=--
]{siunitx}
\usepackage{booktabs}
\usepackage{subcaption}
\usepackage[rightcaption]{sidecap}
\usepackage{algorithm}
\usepackage{algpseudocodex}
\usepackage[capitalise,nameinlink]{cleveref}
\usepackage{stmaryrd}
\usepackage{xcolor}
\usepackage[compat=0.6]{yquant}
\AddToHook{cmd/appendix/before}{
	\crefalias{section}{appendix}
	\crefalias{subsection}{appendix}
}

\newtheorem{theorem}{Theorem}
\newtheorem{cor}{Corollary}
\newtheorem{prop}{Proposition}
\newtheorem{lem}{Lemma}
\theoremstyle{definition}

\newtheorem{eg}{Example}
\newtheorem*{eg*}{\Cref{eg:d3_color_code_acceptance_probability} cont}
\theoremstyle{remark}

\crefname{process}{process}{processes}
\crefname{prop}{Proposition}{Propositions}
\crefname{lem}{Lemma}{Lemmas}
\crefname{eg}{Example}{Examples}

\usetikzlibrary{positioning, arrows.meta, shapes, calc}

\definecolor{lightred}{rgb}{1, 0.7, 0.7}
\definecolor{darkred}{rgb}{0.9, 0, 0}
\definecolor{lightgreen}{rgb}{0.7, 1, 0.7}
\definecolor{darkgreen}{rgb}{0, 0.6, 0}
\definecolor{darkyellow}{rgb}{0.7, 0.7, 0}
\definecolor{zblue}{RGB}{76,166,255}
\tikzset{
    xerror/.style={darkred, ultra thick},
    yerror/.style={darkgreen, ultra thick},
    sxerror/.style={darkyellow, ultra thick},
    every picture/.append style={scale=0.85},
}
\yquantset{
    xerror/.style={style=darkred, control style={ultra thick}},
    yerror/.style={style=darkgreen, control style={ultra thick}},
	dashedcnot/.style={style=darkred, control style={dashed}},
}
\tikzset{
	msc wire above/.style={
		preaction={
			draw=#1,
			line width=2.6pt,
			opacity=0.6,
			transform canvas={yshift=1.3pt},
		},
	},
	msc wire below/.style={
		preaction={
			draw=#1,
			line width=2.6pt,
			opacity=0.6,
			transform canvas={yshift=-1.3pt},
		},
	},
	msc control left/.style={
		preaction={
			draw=#1,
			line width=2.6pt,
			opacity=0.6,
			transform canvas={xshift=-1.3pt},
		},
	},
	msc control right/.style={
		preaction={
			draw=#1,
			line width=2.6pt,
			opacity=0.6,
			transform canvas={xshift=1.3pt},
		},
	},
	msc event/.style={text=black},
}

\makeatletter
\newcommand{\yquantclipfrombackground}[1]{%
  \csletcs{pgf@sh@cliphorz@#1}{pgf@sh@bg@#1}%
}
\makeatother
\yquantclipfrombackground{diamond}
\yquantclipfrombackground{regular polygon}

\newcommand{\materialsaddress}{\affiliation{Department of Materials,
University of Oxford,
Parks Road,
Oxford OX1 3PH,
United Kingdom}}

\newcommand{\engaddress}{\affiliation{Department of Engineering Science, University of Oxford, Parks Road, Oxford OX1 3PJ, United Kingdom}}

\newcommand{\mathaddress}{\affiliation{Mathematical Institute, University of Oxford, Woodstock Road, Oxford OX2 6GG, United Kingdom}}

\newcommand{\qmaddress}{\affiliation{Quantum Motion, 9 Sterling Way, London N7 9HJ, United Kingdom}}

\newcommand{\qiqbaddress}{\affiliation{Center for Quantum Information and Quantum Biology, The University of Osaka, 1-2 Machikaneyama, Toyonaka 560-0043, Japan}}

\newcommand{\impaddress}{\affiliation{Department of Computing, Imperial College London, 180 Queen's Gate, London SW7 2AZ, United Kingdom}}

\begin{document}

\title{Diagnosing and Restoring the Degraded Fault Distance of Magic State Cultivation}
% \title{Hook Clifford Errors Degrade the Fault Distance of Magic State Cultivation}

\author{Tim Chan}
\email{timothy.chan@materials.ox.ac.uk}
\materialsaddress
\qiqbaddress
\orcid{0000-0001-6187-7402}

\author{Armands Strikis}
\mathaddress
\qmaddress

\author{Zhu Sun}
\materialsaddress
\mathaddress
\qmaddress

\author{Zhenyu Cai}
\impaddress
\qmaddress
\engaddress
\email{z.cai1@imperial.ac.uk}

\begin{abstract}
	T-state cultivation is a resource-efficient protocol producing logical T states
	but recent benchmarks show that its logical error rate
	is considerably higher than intended
	i.e.\ than S-state cultivation,
	which is the analogous protocol for producing logical S states.
	In this paper,
	we explain this T--S discrepancy by analytically showing,
	under circuit-level depolarising noise,
	that distance-3 (-5) T-state cultivation has fault distance 2 (3)
	due to Pauli hook errors that propagate to
	coherent Clifford errors after its final double-check circuit;
	such errors remain Pauli in S-state cultivation,
	which consequently retains fault distance 3 (5).
	As part of our analysis we derive a general formula,
	and an $\mathcal O(n^3)$-time algorithm for fixed logical-qubit count,
	for the acceptance probability of
	a logical mixed state afflicted with a Clifford error,
	where $n$ is the physical qubit count.
	We then design flags that detect the malignant hook errors in cultivation,
	improving the pre-escape logical error rate
	from $\upTheta(p^3)$ to $\upTheta(p^5)$.
	At noise level $p =\num{1e-3}$,
	this is a 4.9$\times$ improvement,
	costing a 1.36$\times$ increase in attempts per accepted shot.
\end{abstract}

\maketitle
\begin{figure}[h]
	\centering
	\input{figures/from_yquant/yquant_circuits/2_fault_configuration.tex}
	\hfill
	\includegraphics[width=0.77\linewidth]{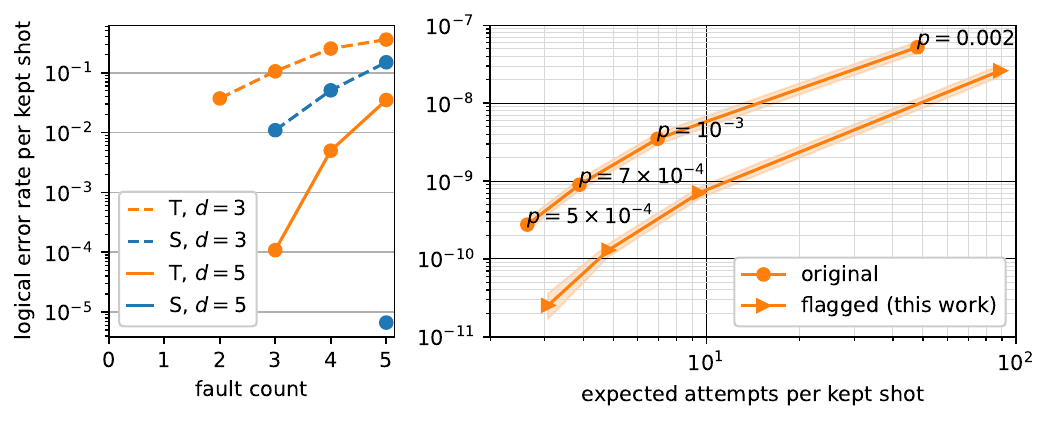}
	\begin{subfigure}[b]{0.24\linewidth}
		\centering
		\caption{In T-state cultivation,
		two Pauli error events can propagate into a weight-3 Clifford error
		that can lead to an undetected logical error.}
		\label{fig:2_fault_configuration}
	\end{subfigure}
	\hfill
	\begin{subfigure}[b]{0.27\linewidth}
		\centering
		\caption{The double check fault distance
		is 2 (3) for distance-3 (-5) T-state cultivation,
		rather than 3 (5) for its S-state proxy.
		Missing datapoints are exactly zero.}
		\label{fig:lep_against_fault_count}
	\end{subfigure}
	\hfill
	\begin{subfigure}[b]{0.43\linewidth}
		\centering
		\caption{In distance-5 T-state cultivation,
		detecting the degrading hook errors with our Z flags
		improves the simulated pre-escape logical error rate per kept shot.
		Shading shows the region within a factor of 1000 of the maximum likelihood.}
		\label{fig:d5_pre_escape_ler}
	\end{subfigure}
	\caption{Summary of results.
	Blue (orange) plot data correspond to cultivating S (T) states.}
	\label{fig:summary}
\end{figure}

\section{Introduction}
\label{sec:introduction}

Magic state cultivation (MSC)~\cite{Gidney2024},
or T-state cultivation,
is a relatively new protocol for producing logical T states
that has recently exploded in popularity
\cite{Chen2026,Claes2025,Sahay2026,Vaknin2025,Wan2025b,Hirano2025a,Rosenfeld2025,Hetenyi2026,Chen2026a}
due to its practicality:
it shrinks the spacetime cost of a logical T state
to that of a lattice surgery CX.
Simulating MSC on a classical computer is costly
due to its non-Clifford nature,
so early numerical benchmarks \cite{Xu2025,Ibe2025,Hines2026}
relied on simulating a Clifford proxy,
S-state cultivation,
on a stabiliser simulator such as Stim~\cite{Gidney2021a}.
Since the T- and S-state cultivation circuits share the same structure,
it was expected that their logical error rates would be similar.
However, simulations from the original paper~\cite{Gidney2024} showed that
distance-3 T-state cultivation in its early stages
has a (roughly 2$\times$) higher logical error rate
than its S-state analogue.
Subsequent papers \cite{Wan2026,Chen2026b} observed that
this T--S discrepancy widens as the noise level decreases.
More recent papers \cite{Li2025,Tuloup2026,Chase2026,Fang2026a}
simulate distance-5 cultivation
and observe an even larger (roughly 11$\times$) T--S discrepancy
that also widens as the noise level decreases;
this has yet to be explained.
This raises concerns for resource estimates that rely on MSC,
and highlights the danger of
using Clifford proxies to simulate
the error behaviour of non-Clifford protocols.

Motivated by this,
we analytically explain the widening T--S discrepancy
with hopes to improve cultivation in the future.
We adopt the approach taken by Daguerre and Kim~\cite[\S V]{Daguerre2025},
which is to consider the Clifford errors that result from
propagating Pauli errors through a layer of non-Clifford gates.
Then,
instead of Monte Carlo sampling,
we exhaustively enumerate the dominant contributions
to the logical error rate,
which directly tells us the \emph{fault distance}
i.e.\ level of logical-error suppression in the physical noise level.
It is known that the fault distance of a subcircuit,
called the \emph{double check},
within distance-3 (-5) S-state cultivation is 3 (5): as expected.
To our knowledge,
we are the first to show that it is instead 2 (3) in T-state cultivation;
see \cref{fig:lep_against_fault_count}.
This is due to the following reason:
the CXs spread Pauli errors which are then rotated
to multiqubit Clifford errors by the physical T-gates in the circuit;
\cref{fig:2_fault_configuration} shows an example.
Rather than merely changing the signs of the code stabilisers,
such an error transforms the code into a different stabiliser code
whose codespace can overlap the original,
allowing postselection to accept a logically erroneous component.
By contrast,
the corresponding errors remain Pauli in S-state cultivation
and are deterministically detected.

To analyse these Clifford errors,
we derive a formula (\cref{theorem:probability_trivial_syndrome})
for the probability to observe a trivial syndrome
for any mixed state encoded in any Pauli stabiliser code afflicted with any Clifford error.
We also provide an algorithm (\cref{alg:probability_trivial_syndrome})
to compute this probability in $\mathcal O(n^3)$ time
for fixed logical-qubit count,
where $n$ is the physical qubit count.
These two general results may be of independent interest.
We then calculate the fidelity (\cref{prop:resulting_fidelity})
specifically for T states afflicted with
the Clifford errors that are encountered in MSC,
conditioned on a trivial syndrome.

Having enumerated the malignant Clifford errors in MSC,
we use the technique of flags~\cite{Chao2018}
to design 2D-local Z-flagged versions of the double checks that detect those errors.
We then use the fast universal simulator SymFT~\cite{Fang2026a}
to simulate the whole MSC protocol except
the final so-called `escape' stage, which is purely Clifford.
Monte Carlo sampling shows these Z flags recover
the intended scaling of logical error rate with noise level,
which suggests hook Clifford errors
are the sole reason for the degraded fault distance of pre-escape MSC.
These sampling results are summarised in \cref{fig:d5_pre_escape_ler}
which shows the practical improvement
these Z flags bring to distance-5 T-state cultivation.
This demonstrates that the fault tolerance of non-Clifford circuits
can be improved by considering the propagation of non-Pauli errors,
even with the simple addition of more Clifford gates.

\Cref{sec:background} covers the prerequisite theory.
We first show the lower-than-expected fault distance in MSC
in \cref{sec:fault_distance_of_magic_state_cultivation},
then present our general Clifford-error analysis in \cref{sec:clifford_error_analysis}.
We quantitatively analyse the fault configurations in \cref{sec:thorough_analysis}
before evaluating our Z-flagged version of MSC in \cref{sec:z_flagged_double_checks}
and concluding in \cref{sec:conclusion}.
\Cref{sec:formal_definitions} provides definitions
required for the rest of the appendices
and \cref{sec:additional_figures} shows alternative figures to those in the main text.
The code for \cref{alg:probability_trivial_syndrome}
and for all the numerics in this paper
is on GitHub at \cite{Chan2026c_quantum_bibstyle}.
\section{Background}
\label{sec:background}
In this paper,
we will use the following notation.
An overline indicates
logical states and operator representatives
e.g.\ $\ket{\overline{0}}, \overline{Z}$.
Physical operators are without overline
and the standard ones are typed upright:
I, X, Y, Z,
$\S =\diag(1, \i)$,
$\T =\diag(1, \e^{\i \uppi /4})$,
CX.
We use the shorthand $\H_\pm :=\frac{1}{\sqrt{2}}(\X \pm \Y)$.
Sets are typed calligraphically e.g.\ $\mathcal{Q}$.
If clear from the context,
we refer to elements from the Pauli (Clifford) group
simply as Paulis (Cliffords).
Given a 1-qubit unitary $U$,
define $U_\mathcal{Q} :=\prod_{q \in \mathcal{Q}} U_q$
for integer set $\mathcal{Q} \subseteq `{1, \dots, n}$,
where $U_q$ denotes $U$ acting on qubit $q$
and $\I$ acting on the other qubits.
The \emph{support} of a multiqubit unitary
is the set of qubits on which it acts something other than $\I$,
e.g., the support of $U_\mathcal{Q}$ is $\mathcal{Q}$.
A \emph{stabiliser circuit} is a quantum circuit comprising
Clifford gates and 1-qubit Z preparations and measurements~\cite{Aaronson2004}.
The parameters of a quantum error correction code are written as
$\llbracket n, k, d \rrbracket$
for physical qubit count $n$,
logical qubit count $k$,
and distance $d$.
In the case distance is unimportant,
we write $\llbracket n, k\rrbracket$.
The logical action of any operator $U$
on a code may be realised by
multiple physically distinct representatives;
throughout, we use $\overline{U}$ to denote
a \emph{specific} representative.

The set
$\mathcal P_k:=\{\I,\X,\Y,\Z\}^{\otimes k}$
of phaseless $k$-qubit Paulis
is a basis for the space of
$2^k \times 2^k$ Hermitian matrices over $\mathbb R$.
So,
the density operator of any $k$-qubit state $\rho$
can be decomposed into this basis:
\begin{equation}
\rho =
2^{-k} \sum_{L \in \mathcal{P}_k} \alpha_L L,
\qquad \alpha_L \in \mathbb R,
\qquad \alpha_\I =1.
\label{eq:pauli_decomposition}
\end{equation}
Here $\alpha_\I$ is obtained
by ensuring $\tr \rho =1$,
and noting $\tr L =2^k$ if $L =\I$ else 0.

Cliffords map Paulis to Paulis.
Any Clifford can be defined up to a global phase
by how it maps the X and Z operators on each input, e.g.:
\begin{equation}\label{eq:h_pm_action}
	\H_\pm: \begin{cases}
		\X &\mapsto \pm \i \X\Z \\
		\Z &\mapsto -\Z.
	\end{cases}
\end{equation}

\subsection{The Hexagonal Colour Code}
Here we review the $\llbracket n, 1, d \rrbracket$ hexagonal colour code
for its relevance to MSC,
and for examples later on in this paper.
For this code,
choose the representatives
$\overline{X} =\X^{\otimes n}$ and $\overline{Z} =\Z^{\otimes n}$.
For the $\llbracket 7, 1, 3 \rrbracket$ instance in particular,
we label the data qubits according to \cref{fig:d3_layout}
and choose the following set $\mathcal G$ of six stabiliser generators:
\begin{equation}
(G_1,\ldots,G_6) =(
\mathrm{XXXIXII},
\mathrm{IXXXIXI},
\mathrm{IIXIXXX},
\mathrm{ZZZIZII},
\mathrm{IZZZIZI},
\mathrm{IIZIZZZ}).
\label{eq:d3_color_stabilizer_generators}
\end{equation}
\begin{figure}[b]
	\begin{subfigure}{0.48\linewidth}
		\centering
		\includegraphics[width=0.55\textwidth]{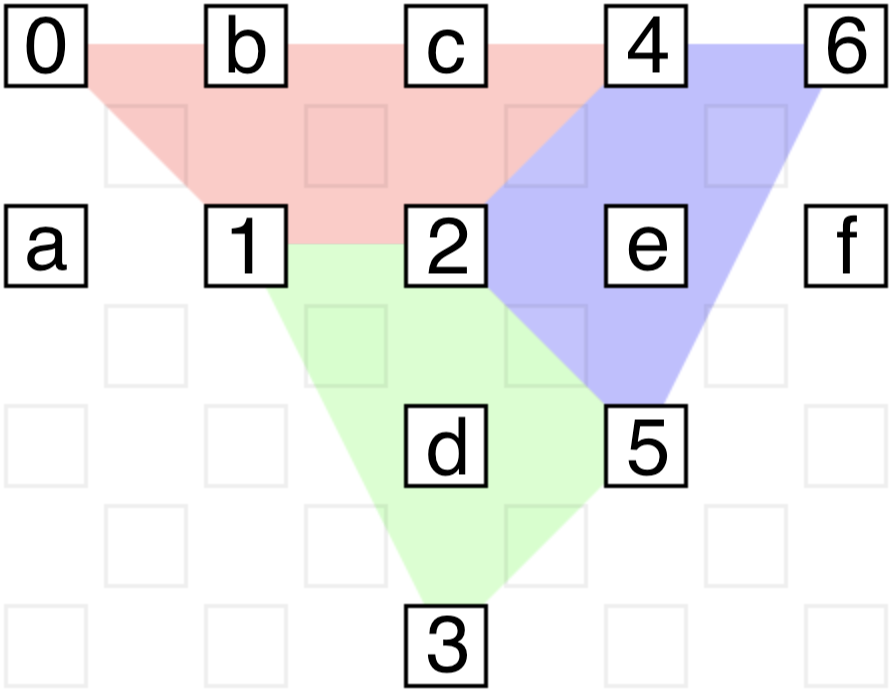}
		\caption{The physical qubit layout used in
		Stage 1 of the magic state cultivation protocol.
		Numbers 0--6 label the data qubits of
		the $\llbracket 7, 1, 3 \rrbracket$ colour code.
		Letters a--f label the ancilla qubits used for
		stabiliser and logical Clifford measurement.}
		\label{fig:d3_layout}
	\end{subfigure}
	\hfill
	\begin{subfigure}{0.48\linewidth}
		\centering
		\includegraphics[width=0.7\textwidth]{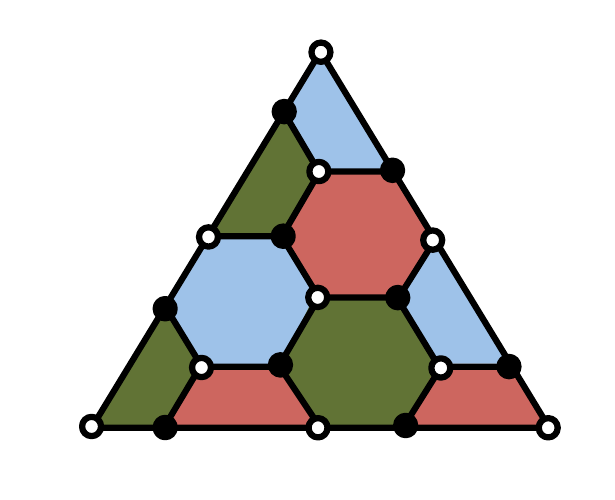}
		\caption{Representing the
		$\llbracket 19, 1, 5 \rrbracket$ colour code
		as a bipartite graph
		between light and dark nodes.
		Figure from \cite{Kubica2015}.}
		\label{fig:bipartite_graph}
	\end{subfigure}
	\caption{Distance-3 and -5 hexagonal colour codes.
	Each colour-shaded polygon
	represents an X- and a Z-type stabiliser generator,
	whose support is the set of qubits at its corners.}
\end{figure}
\noindent
We later use the following.
\begin{lem}\label{lem:logical_H_is_transversal}
On the $\llbracket n, 1, d \rrbracket$ hexagonal colour code,
logical $\H_\pm$ can be implemented transversally as follows.
Represent the physical qubits of the code as nodes,
and its plaquette edges as edges,
on a bipartite graph as the example in \cref{fig:bipartite_graph} shows.
Let $\mathcal{L}$ and $\mathcal{S}$ denote the two parts of the bipartition,
labelled such that $|\mathcal L|-|\mathcal S|\equiv 1 \pmod 4$;
such a labelling is always possible because $n$ is odd.
Then the implementation is
\begin{equation}
\overline{H_\pm} :=(\H_\pm)_\mathcal{L} (\H_\mp)_\mathcal{S}.
\label{eq:logical_H_is_transversal}
\end{equation}
\end{lem}
\begin{proof}
	It suffices to show $\overline{H_\pm}$ preserves the stabiliser group
	and has the desired logical action.
	The code comprises $w$-sided plaquettes where $w \in `{4, 6}$,
	each corresponding to a weight-$w$ stabiliser generator
	supported on an equal number of qubits from $\mathcal{L}$ and $\mathcal{S}$.
	Given this and \cref{eq:h_pm_action},
	$\overline{H_\pm}$ transforms the stabiliser generators as
	$\X^{\otimes w} \mapsto \X^{\otimes w} \Z^{\otimes w}$
	or $\Z^{\otimes w} \mapsto \Z^{\otimes w}$,
	hence preserves the stabiliser group.
	Also,
	$\overline{H_\pm}$ maps the logical X and Z observables in the desired way:
	$\X^{\otimes n}
	\mapsto\pm (-1)^{|\mathcal S|} \i^n \X^{\otimes n} \Z^{\otimes n}
	=\pm \i \X^{\otimes n} \Z^{\otimes n}$
	and $\Z^{\otimes n} \mapsto -\Z^{\otimes n}$.
\end{proof}
\noindent
This bipartition-based construction
was originally used for $\overline{S}$
\cite[\S II.C]{Kubica2015}
(hence provides a proof shortcut using $\overline{H_+} =\e^{-\i \uppi/4} \overline{SX}$),
but note other valid allocations of $\H_+, \H_-$ exist
\cite[Lemma~3]{Tansuwannont2025}.
For the rest of this paper,
$\overline{H_+}$ for hexagonal colour codes will refer to
the implementation in \cref{eq:logical_H_is_transversal}.

\subsection{Fault Distance}\label{sec:fault_distance}
In quantum error correction we naturally deal with
noisy quantum circuits that propagate and manipulate logical information.
These circuits usually contain measurements
that enable the construction of detectors:
a \emph{detector} is a set of Pauli measurements
whose parity is deterministic in the absence of noise.
For a given noisy circuit run,
the \emph{syndrome} is a binary vector $\v s$ whose support
indicates the detectors whose parities differ from expected.
We call $\v s =\v 0$ a \emph{trivial} syndrome.
We define an \emph{error event} as an unintended event
that occurs independently (from other error events)
at a specific location in the circuit
with a probability proportional to
some characteristic noise level $p$ of the system
e.g.\ a Pauli X after an S gate.
A \emph{hook error} is an error event that,
due to a physical entangling gate in the circuit,
spreads to eventually affect more than one data qubit.

Often, different error events
	flip the same set of detectors
	and have the same resultant Pauli effect on the data qubits
		when propagated to the end of the stabiliser circuit.
Whenever this occurs,
we group these equivalent error events together
into the same \emph{fault},
defined abstractly by that detector signature and resultant effect
(we give a proper definition in \cref{sec:faults}).
Then, instead of modelling noise as a large set of error events,
we can model it more efficiently as a smaller set of independent faults.
For those familiar with the \emph{detector error model} formalism~\cite{Gidney2021a}:
a fault is similar to an error mechanism
but carries slightly more information;
we define it this way on purpose to allow
subsequent propagation of the resultant effect through non-Cliffords.
A \emph{$k$-fault configuration} is a set of $k$ faults occurring together;
its \emph{fault count} is its cardinality, $k$.

\emph{Fault distance},
also known as \emph{circuit distance},
generalises the notion of code distance
(a property of the code) to circuits.
It is the minimum number of error events
(equivalently, faults)
needed to flip a logical observable
without flipping any detectors \cite[\S 3.1]{Gidney2024}.
For a subcircuit with quantum inputs,
we define its fault distance assuming error-free inputs
and count only error events occurring within the subcircuit.
A high fault distance is desirable because
as $p \to 0$, and conditioned on a trivial syndrome,
the logical error rate of a circuit of fault distance $d$ is $\upTheta(p^d)$.

\subsection{Magic State Cultivation}
\label{sec:magic_state_cultivation}

\begin{figure}
	\centering
	\begin{tikzpicture}
\begin{yquant}[/yquant/operator/minimum width=0mm]

qubit {} d[3];
nobit a[3];
nobit f;

init {$\ket{\overline{T}}$} (d);
init {$\ket{+}$} a[0];
init {$\ket{0}$} a[1,2];

align -;
cnot a[1]|a[0];
cnot a[2]|a[0];
init {$\ket{0}$} f;
cnot f|a[1];
align -;
cnot f|a[2];
dmeter {Z} f;
discard f;

box {$\H_\pm$} d[0]|a[0];
[operator/separation=-2.5mm]
box {$\H_\pm$} d[1]|a[1];
[operator/separation=-2.5mm]
box {$\H_\pm$} d[2]|a[2];

dmeter {X} a;
discard a;

barrier (-);

init {$\ket{+}$} a[0];
init {$\ket{0}$} a[1,2];

cnot a[1]|a[0];
cnot a[2]|a[0];

box {$\H_\pm$} d[0]|a[0];
[operator/separation=-2.5mm]
box {$\H_\pm$} d[1]|a[1];
[operator/separation=-2.5mm]
box {$\H_\pm$} d[2]|a[2];

cnot a[2]|a[0];
cnot a[1]|a[0];

dmeter {X} a[0];
dmeter {Z} a[1,2];
discard a;

barrier (-);

init {$\ket{+}$} a;

box {$\H_\pm$} d[0]|a[0];
[operator/separation=-2.5mm]
box {$\H_\pm$} d[1]|a[1];
[operator/separation=-2.5mm]
box {$\H_\pm$} d[2]|a[2];

cnot a[2]|a[0];
cnot a[1]|a[0];

dmeter {X} a[0];
discard a[0];
init {$\ket{+}$} a[0];

cnot a[1]|a[0];
cnot a[2]|a[0];

box {$\H_\pm$} d[2]|a[2];
[operator/separation=-2.5mm]
box {$\H_\pm$} d[1]|a[1];
[operator/separation=-2.5mm]
box {$\H_\pm$} d[0]|a[0];

dmeter {X} a;
discard a;

\end{yquant}
\end{tikzpicture}
	\begin{subfigure}{0.25\linewidth}
		\caption{The \emph{short single check}, used in
		Itogawa et al.~\cite{Itogawa2025}.
		This involves preparing a fault-tolerant GHZ state
		and measuring it transversally;
		the $\overline{H_+}$ measurement is the
		parity of these X measurements.}
		\label{fig:short_single_check_sketch}
	\end{subfigure}
	\hfill
	\begin{subfigure}{0.24\linewidth}
		\caption{The \emph{long single check}, used in
		Sahay et al.~\cite{Sahay2026}.
		This involves encoding a GHZ state,
		unencoding it,
		and measuring it;
		the $\overline{H_+}$ measurement is the
		single physical X measurement.}
		\label{fig:long_single_check_sketch}
	\end{subfigure}
	\hfill
	\begin{subfigure}{0.45\linewidth}
		\caption{The \emph{double check}, used in
		Gidney et al.~\cite{Gidney2024}.
		This is conceptually (a) preceded by its time reverse,
		and constitutes two $\overline{H_+}$ measurements:
		the first is the X measurement midway;
		the second is the parity of the X measurements at the end.}
		\label{fig:double_check_sketch}
	\end{subfigure}
	\caption{A sketch of three different methods to measure
	$\overline{H_+}$ from \cref{lem:logical_H_is_transversal}.
	The top three qubits represent the $n >3$ data qubits of the colour code.
	The size of the GHZ state in each method is flexible~\cite[\S B.3.b]{Sahay2026}.}
	\label{fig:check_options}
\end{figure}
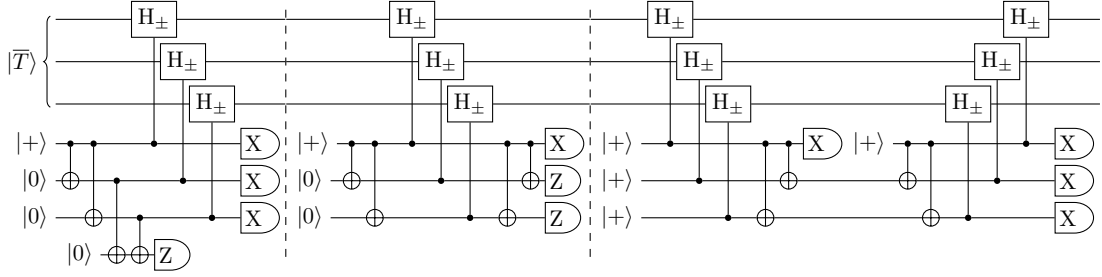
There are two variants of MSC:
MSC-3 and MSC-5,
named after the distance of the largest colour code used.
Both variants comprise three stages:
injection, cultivation, and escape.

We first describe MSC-3.
The magic state $\ket{\T} \propto \ket{0} +\e^{\i \uppi/4} \ket{1}$
is the state stabilised by the Clifford $\H_+$,
so $\ket{\overline{T}}$ is stabilised by $\overline{H_+}$.
\begin{description}
	\item[Stage 1] (injection) prepares a $\ket{\overline{T}}$ state
	in the $\llbracket 7, 1, 3 \rrbracket$ colour code
	using a unitary circuit of fault distance 1,
	followed by one stabiliser measurement round.
	\item[Stage 2] (cultivation) measures the observable $\overline{H_+}$ twice;
	out of the possible methods outlined in \cref{fig:check_options} for this,
	MSC uses the \emph{double check} in \cref{fig:double_check_sketch}
	as it is the most amenable to 2D nearest-neighbour connectivity.
	Its explicit form for MSC-3,
	which we call the \emph{distance-3 double check},
	is shown in \cref{fig:d3a6};
	this relies on the decompositions
	$\H_+ =\T\X\T^\dag$ and
	$\H_- =\T^\dag \X\T$,
	and has two further optimisations:
	$\T^\dag$ cancels $\T$ between the two applications of controlled-$\H_\pm$,
	and the GHZ-state (un)encoding is morphed onto some data qubits
	for connectivity reasons.
	Note the circuit in the original paper~\cite[Figure~7]{Gidney2024} mistakenly
	measures $\H_+^{\otimes 7}$ which is logically equivalent to $\overline{H_-}$.
	\item[Stage 3] (escape) converts the colour code into a much larger surface code.
	Throughout the whole protocol up until this point:
	any detector flip causes the attempt to be aborted.
	Otherwise, after conversion,
	some number of stabiliser measurement rounds are performed;
	the syndrome is given to a decoder and the state is accepted
	if the confidence of the decoder exceeds some preset threshold.
\end{description}
MSC-3 is summarised in \cref{fig:stage_flowchart},
and is intended to have fault distance 3 before its escape stage.
\begin{SCfigure}
	\includegraphics{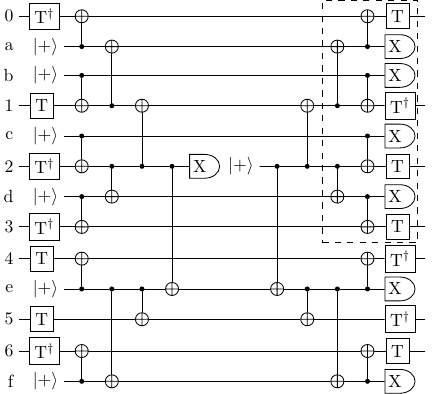}
	\caption{The double-check circuit of distance-3 T-state cultivation.
	The qubits are labelled according to \cref{fig:d3_layout}.
	Each X measurement corresponds to a detector.
	The dashed box is the region shown in \cref{fig:2_fault_configuration}.
	For S-state cultivation, replace T with S;
	click
	\href{https://algassert.com/crumble\#circuit=Q(0,0)0;Q(0,1)1;Q(1,0)2;Q(1,1)3;Q(2,0)4;Q(2,1)5;Q(2,2)6;Q(2,3)7;Q(3,0)8;Q(3,1)9;Q(3,2)10;Q(4,0)11;Q(4,1)12;POLYGON(0,0,1,0.25)8_11_10_5;POLYGON(0,1,0,0.25)7_10_5_3;POLYGON(1,0,0,0.25)3_5_8_0;TICK;MPP_Y0*Y8*Y11*Y10*Y7*Y5*Y3;TICK;MPP_X8*X0*X3*X5;TICK;MPP_X3*X5*X10*X7;TICK;MPP_X11*X8*X5*X10;TICK;MPP_Z8*Z0*Z3*Z5;TICK;MPP_Z7*Z10*Z5*Z3;TICK;MPP_Z11*Z8*Z5*Z10;TICK;RX_12_9_4_2_6_1;S_3_8_10;S_DAG_0_5_7_11;TICK;CX_1_0_9_8_6_7_4_5_2_3_12_11;TICK;CX_3_1_5_6_9_12;TICK;CX_5_3_9_10;TICK;CX_5_9;TICK;MX_5;DT(0,0,0)rec[-1]_rec[-8];TICK;RX_5;TICK;CX_5_9;TICK;CX_5_3_9_10;MARKX(0)3;TICK;CX_3_1_5_6_9_12;TICK;CX_1_0_9_8_6_7_12_11_4_5_2_3;MARKY(0)7;TICK;MX_12_9_4_2_6_1;S_5_11_7_0;S_DAG_3_8_10;DT(4,1,1)rec[-6];DT(3,1,1)rec[-5];DT(2,1,1)rec[-4]_rec[-7];DT(1,0,1)rec[-3];DT(2,2,1)rec[-2];DT(0,1,1)rec[-1];TICK;MPP_X8*X0*X3*X5;DT(3,0,2)rec[-1]_rec[-8]_rec[-14];TICK;MPP_X7*X10*X5*X3;DT(1,1,3)rec[-1]_rec[-9]_rec[-14];TICK;MPP_X11*X8*X5*X10;DT(4,0,4)rec[-1]_rec[-10]_rec[-14];TICK;MPP_Z8*Z0*Z3*Z5;DT(3,0,5)rec[-1]_rec[-14];TICK;MPP_Z7*Z10*Z5*Z3;DT(2,3,6)rec[-1]_rec[-14];TICK;MPP_Z11*Z8*Z5*Z10;DT(4,0,7)rec[-1]_rec[-14];TICK;MPP_Y0*Y8*Y11*Y10*Y7*Y5*Y3;OI(0)rec[-1]_rec[-8]_rec[-9]_rec[-12]_rec[-13]}{this Crumble link}
	to see that circuit with the error events from \cref{fig:2_fault_configuration}.}
	\label{fig:d3a6}
\end{SCfigure}

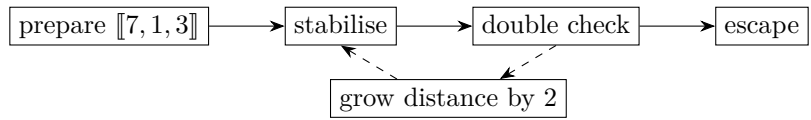
\begin{figure}[H]
	\centering
	\begin{tikzpicture}[
		stage/.style={
			draw=black,
			text centered,
			text depth=.25ex,
			text height=1.5ex
		},
		arrow/.style={-{Stealth[length=2mm]}},
	]
	% nodes
	\node[stage] (prepare) {prepare $\llbracket 7, 1, 3 \rrbracket$};
	\node[stage] (stabilise) [right=of prepare] {stabilise};
	\node[stage] (double_check) [right=of stabilise] {double check};
	\node[stage] (escape) [right=of double_check] {escape};
	\node[stage, yshift=3mm] (grow) [below=of $(stabilise)!0.5!(double_check)$] {grow distance by 2};
	% lines
	\draw[arrow] (prepare.east) -- (stabilise.west);
	\draw[arrow] (stabilise.east) -- (double_check.west);
	\draw[arrow] (double_check.east) -- (escape.west);
	\draw[arrow,dashed] (double_check.south) -- ([xshift=8mm]grow.north);
	\draw[arrow,dashed] ([xshift=-8mm]grow.north) -- (stabilise.south);
	\end{tikzpicture}
	\caption{Summary of magic state cultivation.
	Distance-5 cultivation traverses the dashed path once.}
	\label{fig:stage_flowchart}
\end{figure}

MSC-5,
also summarised in \cref{fig:stage_flowchart},
is the same as MSC-3 through to the distance-3 double check,
after which the $\llbracket 7, 1, 3 \rrbracket$ colour code
is grown into a $\llbracket 19, 1, 5\rrbracket$ colour code.
Three stabiliser measurement rounds are performed,
then $\overline{H_+}$ is again measured twice
using the distance-5 analogue of \cref{fig:d3a6};
we call this the \emph{distance-5 double check}.
Stage 3 is the same as for MSC-3---with
larger initial and final codes.
MSC-5 is intended to have fault distance 4 before its escape stage;
this increases to 5 if four instead of three stabiliser rounds
are performed before the distance-5 double check~\cite[\S 2.3]{Gidney2024} \cite[\S 4]{Hetenyi2026}.
Gidney et al.~\cite{Gidney2024} chose three instead of four rounds
because the spacetime cost of adding another round is not worth
the small improvement in logical error rate at practically relevant noise levels.

The Clifford analogue of T-state cultivation
is S-state cultivation,
whose circuit can be obtained by replacing all
T ($\T^\dag$) gates with S ($\S^\dag$).
As \cref{sec:introduction} mentions,
multiple papers~\cite{Gidney2024,Li2025,Tuloup2026,Chase2026,Fang2026a,Chen2026b}
report an unexplained difference in
the pre-escape logical error rate per kept shot
between T- and S-state cultivation,
which increases as the noise level decreases.
\section{Motivating Example: Degraded Fault Distance in Cultivation}
\label{sec:fault_distance_of_magic_state_cultivation}
The reason for the widening T--S discrepancy is that for S-state cultivation,
the double checks have the fault distance they were designed for.
For T-state cultivation,
the final distance-3 (-5) double check
instead has fault distance 2 (3);
this means that the logical error rate per kept shot
scales as $p^2$ ($p^3$) instead of $p^3$ ($p^5$), for noise level $p$.
For the rest of this paper we thus focus on the final double check.
The following shows an explicit example of
two error events that can lead to a logical error
in distance-3 T-state cultivation.

\begin{eg}\label{eg:d3_color_code_acceptance_probability}
Consider two 1-qubit error events, each occurring with probability $\upTheta(p)$;
namely, an $\X_1$ and a $\Y_3$ error event
at the locations shown in \cref{fig:2_fault_configuration}.
The $\X_1$ is a hook error:
it propagates through two CXs to become $\X_{0 \a 1}$,
but ancilla qubit $\a$ is measured in the X-basis
so the X error on it will have no effect.
This leaves us with $\X_{0 1} \Y_3$ immediately before the T-gates,
which will go undetected past this circuit.
Note that the data qubits $(0, 1, 3)$ support a logical Pauli;
we will omit their subscripts in what follows.
Pushing any unitary $U$ through a gate $G$ transforms $U$ into $GUG^\dag$;
now consider pushing $\X\X\Y$ through the given T-gates:
\begin{align}
	E
	&:=
	(\T \T^\dag \T)
	(\X \X \Y)
	(\T \T^\dag \T)^\dag \notag \\
	&=-(\H_+)_0 (\H_-)_{1 3}.
	\label{eq:clifford_error_in_superposition}
\end{align}
The theory we develop in the next section,
together with \cref{lem:insert_projector},
shows that
the Clifford error $E$ flips no subsequent detectors
with probability of 1/4 and,
if so,
leads to an undetected logical error.
An intuitive explanation is:
$E$ completely randomises two X-type stabiliser generators of the code,
so the probability of both returning to their +1 eigenvalue is $(1/2)^2$.
As for the logical information,
$E$ anticommutes with the stabiliser $\overline{H_+}$
that defines the logical state,
so it rotates the logical state to the orthogonal one.
\end{eg}

To contrast,
we briefly show why the error $\X_{0 1} \Y_3$
never leads to logical error in the analogous circuit for S-state cultivation.
\begin{eg}
\label{eg:d3_color_code_s_state}
In S-state cultivation,
the T-gates are replaced by S-gates.
Consider pushing $\X\X\Y$ through these S-gates:
\begin{equation}
	(\S\S^\dag \S) (\X\X\Y)(\S\S^\dag \S)^\dag
	=\Y_{0 1} \X_3.
\end{equation}
This Pauli error anticommutes with $G_2$
from the set of stabiliser generators
\cref{eq:d3_color_stabilizer_generators};
hence, it is deterministically detected.
\end{eg}

The next section is motivated by
the question posed at the end of
\cref{eg:d3_color_code_acceptance_probability}
(\textit{how often does the Clifford error evade detection
and end up as an accepted logical error?});
unfortunately however,
the final double check is followed immediately by the escape stage,
whose details are quite involved.
We avoid propagating the Clifford error $E$ through this stage
using our lemma in \cref{sec:reducing_code_conversion_to_codespace_projection},
which allows us to insert
an imaginary noiseless stabiliser measurement round beforehand
without affecting the post-escape state.
We can then analyse the effect of $E$ on that imaginary round instead,
which is the topic of the next section.
Then, in \cref{sec:thorough_analysis},
we resume our analysis of MSC
by enumerating the low-weight logical errors
in MSC-3 and -5.
\section{Clifford-Error Analysis}
\label{sec:clifford_error_analysis}

In this section,
we consider in a general context how physical Clifford errors
affect logical states.
Specifically,
we consider afflicting a physical Clifford error on an encoded state,
then noiselessly measuring the stabiliser generators
(which we assume to be Paulis)
of the code.
Since the error in general is not Pauli,
these noiseless measurements in general are nondeterministic.
We may get a trivial syndrome ($\v s =\v 0$),
in which case the afflicted state is projected back into the codespace.
We are interested in two quantities:
(1) the probability to observe this trivial syndrome,
and (2) the fidelity of the resulting state given a trivial syndrome.

Conceptually,
our analysis uses that since Cliffords preserve commutation,
the Clifford error will transform the code
into a slightly different stabiliser code.
Calculating (1) involves counting
how many new stabilisers overlap with the old ones;
the time complexity of this is
polynomial in physical qubit count.
As for (2),
the result depends on the overlap between
the original projector defining the logical state
and the transformed projector.
\Cref{sec:acceptance_probability} answers (1) generally,
and \cref{sec:resulting_fidelity} answers (2)
in the case of T-state encoded in the colour code.

\subsection{Acceptance Probability}
\label{sec:acceptance_probability}
Here we present a general expression for
the probability $\pr(\v s =\v 0)$ to observe a trivial syndrome.
\begin{theorem}
\label{theorem:probability_trivial_syndrome}
Consider an arbitrary $k$-qubit mixed state $\rho$
encoded in an $\llbracket n, k\rrbracket$ Pauli stabiliser code,
which then suffers a physical Clifford error $E$.
Let $\mathcal G$ be a set of $n-k$ independent stabiliser
generators, which we measure noiselessly after the error.
The probability of a trivial syndrome is
\begin{equation}
\label{eq:probability_trivial_syndrome}
\pr(\v s =\v 0)
=\begin{cases*}
0, & if $`<E^\dagger \mathcal G E>
\cap \big(-`<\mathcal G>\big)
\neq \varnothing$, \\
2^{-s} \sum_{L} \omega_L \alpha_L,
& else,
\end{cases*}
\end{equation}
where the coefficients $\{\alpha_L\}_L$ define $\rho$ by
\cref{eq:pauli_decomposition},
and $s :=\rk`< \mathcal G\cup E^\dagger\mathcal G E > -(n -k)$ is
the number of new constraints imposed by the transformed stabilisers.
The sum is over all $L\in\mathcal P_k$ whose logical representative satisfies
\begin{equation}
\label{eq:logical_membership}
\overline L
\in \omega_L
`<E^\dagger \mathcal G E>
\langle\mathcal G\rangle
\end{equation}
for some sign $\omega_L\in\{\pm1\}$;
in the nonzero case, this sign is unique.
\end{theorem}

% \blue{
% \begin{theorem}
% Consider an arbitrary $k$-qubit mixed state $\rho$
% encoded in an $\llbracket n,k\rrbracket$ Pauli stabiliser code,
% which then suffers a physical Clifford error $E$.
% Let $\mathcal G$ be a set of $n-k$ independent stabiliser
% generators, which we measure noiselessly after the error.
% The probability of a trivial syndrome is
% \begin{equation}
%     \label{eq:probability_trivial_syndrome_zero}
%     \pr(\v s=\v 0)=0,
%     \qquad \text{if}\quad
%     \langle E^\dagger\mathcal G E\rangle
%     \cap \bigl(-\langle\mathcal G\rangle\bigr)
%     \neq \varnothing.
% \end{equation}

% Otherwise, let $\mathcal L_E\subseteq\mathcal P_k$ contain
% precisely those logical Paulis $L$ for which there exist
% a physical Pauli representative $\overline L$ of $L$
% and a sign $\omega_L\in\{+1,-1\}$ satisfying
% \begin{equation}
%     \label{eq:logical_membership}
%     \omega_L\overline L
%     \in \langle E^\dagger\mathcal G E\rangle.
% \end{equation}
% Then
% \begin{equation}
%     \label{eq:probability_trivial_syndrome}
%     \pr(\v s=\v 0)
%     =
%     2^{-s}
%     \sum_{L\in\mathcal L_E}\omega_L\alpha_L,
% \end{equation}
% where $\{\alpha_L\}_L$ are the Pauli coefficients of $\rho$
% defined in \cref{eq:pauli_decomposition}, and
% \[
%     s :=
%     \rk\langle\mathcal G\cup E^\dagger\mathcal G E\rangle
%     -(n-k)
% \]
% is the number of additional independent Pauli generators
% contributed by the transformed checks beyond the original
% stabiliser group.
% \end{theorem}
% Note that for each $L$ satisfying \cref{eq:logical_membership}, the sign $\omega_L$ is unique.
% }

\noindent
In words:
if the transformed stabiliser group contains
the negative of an original stabiliser,
the probability vanishes.
Otherwise,
a given logical Pauli component contributes to the probability if
any representative is a transformed stabiliser up to a sign.
We can interpret $2^{-s}$ as the acceptance probability for a maximally
mixed logical input, while the sum accounts for dependence on the
logical state.
We prove this statement in \cref{sec:clifford_error_analysis_proofs}
and present an $\mathcal O(n^3)$-time algorithm
to evaluate \cref{eq:probability_trivial_syndrome}
in \cref{sec:acceptance_probability_algorithm}.
In most cases considered in this paper,
the only logical that satisfies \cref{eq:logical_membership}
is identity.
Other logicals can contribute however,
if e.g.\ $E$ maps a stabiliser directly to that logical.

In the special case $E$ is a Pauli then either:
	$E$ flips at least one stabiliser in $\mathcal G$,
	in which case $\pr(\v s =\v 0) =0$,
	or $E^\dag \mathcal G E =\mathcal G$,
	in which case $s =0$ and only $L =\I$ satisfies
	\cref{eq:logical_membership} with $\omega_\I =1$,
	so $\pr(\v s =\v 0) =1$.

\begin{eg*}
To recap,
we have a T-state,
i.e.\ the state stabilised by $\H_+$,
encoded in the $\llbracket 7, 1, 3 \rrbracket$ hexagonal colour code,
which then suffers the Clifford error
$E =-(\H_+)_0 (\H_-)_{1 3}$.
$E$ transforms only $G_1$ and $G_2$ out of $\mathcal G$;
this can be seen from the support of $E$
and by using \cref{eq:h_pm_action,eq:d3_color_stabilizer_generators}.
Since $E=E^\dagger$,
\begin{equation}
(E^\dagger G_1E,\ldots,E^\dagger G_6E)=(
-\mathrm{YYXIXII},
\mathrm{IYXYIXI},
G_3,G_4,G_5,G_6).
\end{equation}
One checks that
$\langle E^\dagger \mathcal G E\rangle
\cap
\bigl(-\langle\mathcal G\rangle\bigr)=\varnothing$
and
$s=\rk\langle\mathcal G\cup E^\dagger\mathcal G E\rangle-n +k
=8-7+1=2$.
Also,
\begin{equation}
\proj{\T}
=\frac12\left(\I+\frac{\X+\Y}{\sqrt2}\right).
\end{equation}
Among the nonzero Pauli-basis coefficients,
only $L=\I$ satisfies \cref{eq:logical_membership},
with $\omega_\I=1$.
Therefore \cref{theorem:probability_trivial_syndrome}
gives $\pr(\v s=\v0) =2^{-s} \omega_\I \alpha_\I =1/4$.
\end{eg*}
\noindent
As an aside,
we recover the following result
that was implied by Aaronson and Gottesman~\cite[p~5]{Aaronson2004}.
\begin{cor}
The fidelity between two pure stabiliser states
whose Pauli stabiliser groups are
$\mathcal A$ and $\mathcal B$ respectively,
is $0$ if $\mathcal B \cap
(-\mathcal A)
\neq \varnothing$,
else $2^{-s}$
where $s := \rk\langle \mathcal A\cup \mathcal B \rangle -n$.
\end{cor}
\begin{proof}
Take $k =0$
in \cref{theorem:probability_trivial_syndrome},
and let $\mathcal A =`<\mathcal G>$
and $\mathcal B =`<E^\dag \mathcal G E>$
(i.e.\ the physical state stabilised by $\mathcal B$
is that of $\mathcal A$,
transformed by the Clifford $E^\dag$);
then the fidelity is $\pr(\v s =\v 0)$.
The density operator of the 0-qubit state is unity
i.e.\ $\alpha_1 =1$,
so \cref{eq:logical_membership}
reduces to $\overline{1} =\I \in \omega_1 \mathcal B \mathcal A$.
If $\omega_1 =-1$ then we are in the first branch of
\cref{eq:probability_trivial_syndrome} and the fidelity vanishes.
Else,
we are in the second branch and
$\sum_L \omega_L \alpha_L =\omega_1 \alpha_1 =1$.
\end{proof}

\subsection{Resulting Fidelity}
\label{sec:resulting_fidelity}
To answer the second question
(\textit{given a trivial syndrome,
what is the fidelity of the resulting state?}),
we consider the more specific scenario,
which we also prove in \cref{sec:clifford_error_analysis_proofs}:

\begin{prop}\label{prop:resulting_fidelity}
Consider a T-state encoded in an $\llbracket n, 1, d \rrbracket$
hexagonal colour code,
which then suffers a tensor product $E$ of Clifford errors from $`{\I, \H_\pm, \Z}$
up to global phase.
The operators $\H_+$, $\H_-$, and $\Z$
square to $\I$ and pairwise anticommute,
so $E$ either commutes or anticommutes with
$\overline{H_+}$ defined in \cref{eq:logical_H_is_transversal}.
If we noiselessly measure a minimal set $\mathcal G$ of stabiliser generators
and observe a trivial syndrome,
the fidelity of the resulting state is
one (zero) if $E$ commutes (anticommutes) with
$\overline{H_+}$;
determining this takes $\mathcal O(n)$ time.
\end{prop}

\begin{eg*}
For $E =-(\H_+)_0(\H_-)_{1 3}$,
the only tensor factor that anticommutes with the corresponding
factor of $\overline{H_+}=(\H_+)_{0 2 3 6}(\H_-)_{1 4 5}$
is at qubit 3.
Thus $E$ anticommutes with $\overline{H_+}$, so the fidelity is zero.
\end{eg*}

Lastly,
we note there is an alternative analysis
(to \cref{theorem:probability_trivial_syndrome,prop:resulting_fidelity})
that treats the Clifford error as a superposition of Paulis
instead of focusing on how it transforms the stabilisers.
This recovers the same results
and is potentially more intuitive,
but is slower:
computing $\pr(\v s =\v 0)$ in $\mathcal O(2^n)$ time.
In \cref{sec:alternative_clifford_error_analysis}
we walk through \cref{eg:d3_color_code_acceptance_probability}
using this alternative analysis.
\begin{SCfigure}[][b]
	\includegraphics[width=0.5\linewidth]{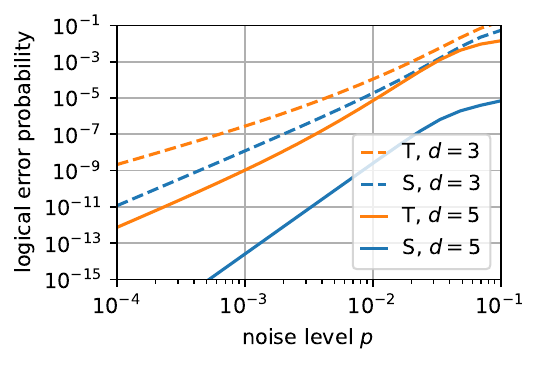}
	\caption{The logical error rate per kept shot of the double-check circuit,
	plotted against noise level $p$.
	Dashed (solid) lines correspond to distance-3 (-5) cultivation,
	and blue (orange) data correspond to cultivating S (T) states.
	These curves are analytic and correct through fifth order in $p$.
	From top to bottom, the gradient of the lines as $p \to 0$ are: 2, 3, 3, 5.
	This figure and \cref{fig:lep_against_fault_count} are derived from the same enumeration data.}
	\label{fig:ler_against_per}
\end{SCfigure}

\section{Fault Enumeration of Magic State Cultivation}\label{sec:thorough_analysis}
We resume our focus on MSC,
combining \cref{lem:insert_projector,sec:clifford_error_analysis}
to find fault configurations in the final double check
with some probability of being undetected;
that is,
not flipping any detectors during and after the double check.
We use the same uniform depolarising noise model as
in the original MSC paper~\cite{Gidney2024},
which is summarised in \cref{sec:noise_model}.
We enumerate all such configurations of \num{\le 5} faults
for both distances 3 and 5;
the latter took 15 seconds on a 2021 Apple MacBook Pro M1.
Compared to Monte Carlo sampling,
enumeration gives us more insight
and allows us to analytically reconstruct
an approximation of the logical error rate not at one noise level,
but as a function of the noise level.
We need only compute this function once
before we can instantaneously probe arbitrarily low logical error rates
and generate plots such as \cref{fig:ler_against_per}.
We note this approach is similar to SyQMA~\cite{Umbrarescu2026},
which can compute the exact symbolic logical error rate
(and thus fault distance) of non-Clifford circuits,
but does so by simulating the noisy state rather than enumerating fault configurations.

We say an undetected fault configuration is
benign (malignant)
if it leads to a logical fidelity of 1 (0).
For MSC-3,
there are 4 malignant 2-fault configurations;
these are explicitly shown in \cref{fig:all_2_fault_configurations}.
\begin{SCfigure}
	\input{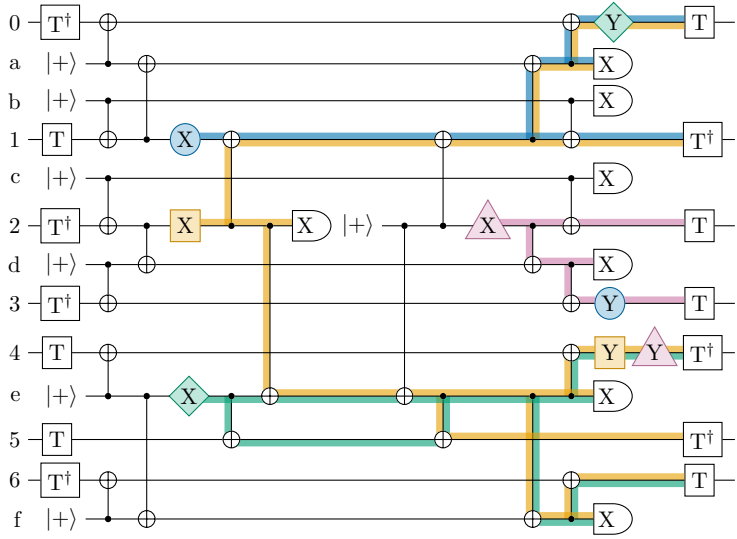}
	\caption{All four 2-fault configurations
	for the distance-3 double-check circuit (\cref{fig:d3a6})
	that can lead to logical error.
	Each configuration is shown as a pair of error events
	of a given shape and colour
	e.g.\ the blue circled X and Y realise
	the configuration in \cref{fig:2_fault_configuration}.
	Note each configuration can be realised by
	multiple equivalent pairs of error events;
	here we choose to show the earliest possible events.
	We also highlight how each X event propagates until the layer of T-gates.}
	\label{fig:all_2_fault_configurations}
\end{SCfigure}
Note these configurations degrade the fault distance of MSC-3 only:
they are caught by the subsequent distance-5 double check in MSC-5.
For MSC-5,
there are 3 (601) malignant 3- (4-) fault configurations
in its final, distance-5 double check;
each configuration contains at least one X hook error.
\Cref{fig:fcc_against_fault_count} shows
the complete data on configuration counts,
from which \cref{fig:lep_against_fault_count} is derived.
In \cref{sec:logical_error_rate_per_kept_shot} we describe how to
analytically construct the logical error rate per kept shot
from these fault enumerations,
the result of which is shown in \cref{fig:ler_against_per}.
The T--S logical error rate discrepancy
observed for both code distances in this subfigure
can be explained by \cref{fig:lep_against_fault_count}.
The first reason is the difference in fault distance
(indicated by the first nonzero point for each line);
this dictates the asymptotic scaling
of the logical error rate as the noise level $p \to 0$.
The second reason is that T-state cultivation has
a higher fraction of malignant configurations
compared to S-state cultivation;
this manifests as a higher constant factor
multiplying the logical error rate.

This degraded fault distance is experienced also by
the distance-3 double check used in
the experimental realisation of MSC~\cite{Rosenfeld2025}
(Crumble link \href{https://algassert.com/crumble\#circuit=Q(0,1)0;Q(0,2)1;Q(0.5,0.5)2;Q(0.5,1.5)3;Q(1,0)4;Q(1,1)5;Q(1,2)6;Q(1.5,0.5)7;Q(1.5,1.5)8;Q(2,1)9;Q(2,2)10;POLYGON(0,0,1,0.25)0_4_9_5;POLYGON(0,1,0,0.25)0_5_6_1;POLYGON(1,0,0,0.25)5_9_10_6;TICK;MPP_Y0*Y4*Y5*Y9*Y10*Y6*Y1;TICK;MPP_X5*X9*X10*X6;TICK;MPP_X0*X5*X6*X1;TICK;MPP_X4*X9*X5*X0;TICK;MPP_Z5*Z9*Z10*Z6;TICK;MPP_Z0*Z5*Z6*Z1;TICK;MPP_Z4*Z9*Z5*Z0;TICK;RX_8_2_7_3;TICK;S_0_9_6;S_DAG_4_5_1_10;TICK;CX_3_0_7_4_2_5_8_10;TICK;CX_5_7_8_9_3_1;TICK;CX_5_3_8_6;TICK;CX_8_5;TICK;MX_8;DT(0,1,0)rec[-1]_rec[-8];TICK;RX_8;TICK;CX_8_5;TICK;CX_5_3_8_6;TICK;CX_5_7_8_9_3_1;TICK;CX_3_0_7_4_2_5_8_10;TICK;S_4_5_1_10;S_DAG_0_9_6;TICK;MX_8_2_7_3;DT(1.5,1.5,1)rec[-4]_rec[-5];DT(0.5,0.5,1)rec[-3];DT(1.5,0.5,1)rec[-2];DT(0.5,1.5,1)rec[-1];TICK;MPP_X5*X9*X10*X6;DT(1,1,2)rec[-1]_rec[-12];TICK;MPP_X0*X5*X6*X1;DT(0,1,3)rec[-1]_rec[-12];TICK;MPP_X4*X9*X5*X0;DT(1,0,4)rec[-1]_rec[-12];TICK;MPP_Z5*Z9*Z10*Z6;DT(1,1,5)rec[-1]_rec[-12];TICK;MPP_Z0*Z5*Z6*Z1;DT(0,1,6)rec[-1]_rec[-12];TICK;MPP_Z4*Z9*Z5*Z0;DT(1,0,7)rec[-1]_rec[-12];TICK;MPP_Y0*Y4*Y5*Y9*Y10*Y6*Y1;OI(0)rec[-1]_rec[-8]_rec[-9]_rec[-11]}{here}):
applying the same analysis to this slightly smaller circuit,
we find it also has fault distance 2.

\section{Z-Flagged Double Checks}\label{sec:z_flagged_double_checks}

So how can the intended fault distance be restored?
Note the distance-3 double check uses 6 ancillas
to interact with the 7 data qubits
so there is necessarily a pair of CXs
from the same ancilla to different data qubits (e to 5, and e to 6 via f).
One may therefore think using 7 ancillas will restore fault distance;
however,
none of the X hook errors in \cref{fig:all_2_fault_configurations}
are due to this CX pair.
Varying ancilla count will not restore fault distance.

A solution which works is to replace the double check
with two short or long single checks,
sketched in \cref{fig:short_single_check_sketch,fig:long_single_check_sketch}
respectively.
We find
using the analysis of the previous section
that these $\overline{H_+}$-measurement methods
do not have degraded fault distance,
at least for code distance 3;
specifically, we analyse the distance-3 short single check
(Crumble link \href{https://algassert.com/crumble\#circuit=Q(0,0)0;Q(0,1)1;Q(0,2)2;Q(0,3)3;Q(0,4)4;Q(1,0)5;Q(1,1)6;Q(1,3)7;Q(1,4)8;Q(2,0)9;Q(2,1)10;Q(2,2)11;Q(2,3)12;Q(2,4)13;Q(3,0)14;Q(3,1)15;Q(3,2)16;Q(3,3)17;Q(3,4)18;Q(4,0)19;Q(4,1)20;Q(4,2)21;Q(4,3)22;Q(4,4)23;MPP_Y6*Y10*Y15*Y11*Y17*Y12*Y7;POLYGON(0,0,1,0.25)7_6_11_12;POLYGON(0,1,0,0.25)10_12_17;POLYGON(1,0,0,0.25)6_11_15;TICK;MPP_X6*X11*X12*X7;TICK;MPP_X6*X10*X15*X11;TICK;MPP_X11*X10*X17*X12;TICK;MPP_Z6*Z11*Z12*Z7;TICK;MPP_Z6*Z10*Z15*Z11;TICK;MPP_Z11*Z10*Z17*Z12;TICK;R_23;RX_22;TICK;CX_22_23;R_21_18;TICK;CX_22_21_23_18;R_13_20;TICK;CX_18_13_21_20;MX_22_23;R_8_19;TICK;CX_13_8_20_19;R_4_14;TICK;CX_8_4_19_14;MX_20;R_9_3;TICK;CX_14_9_4_3;MX_8_19;R_5_2;TICK;CX_3_2_9_5;MX_4;R_0_1_16;TICK;CX_5_0_2_1_21_16;TICK;CX_1_0;MX_5_2_21;S_6_10_12;S_DAG_7_11_17_15;TICK;CX_1_6_3_7_13_12_18_17_9_10_14_15_16_11;M_0;DT(0,0,0)rec[-1];TICK;MX_1_9_14_16_3_13_18;S_11_7_17_15;S_DAG_10_6_12;DT(1,1,1)rec[-1]_rec[-2]_rec[-3]_rec[-4]_rec[-5]_rec[-6]_rec[-7]_rec[-9]_rec[-10]_rec[-11]_rec[-12]_rec[-13]_rec[-14]_rec[-15]_rec[-16]_rec[-17]_rec[-24];TICK;MPP_X6*X11*X12*X7;DT(1,1,2)rec[-1]_rec[-24];TICK;MPP_X6*X10*X15*X11;DT(1,1,3)rec[-1]_rec[-24];TICK;MPP_X11*X10*X17*X12;DT(2,2,4)rec[-1]_rec[-11]_rec[-24];TICK;MPP_Z6*Z11*Z12*Z7;DT(1,1,5)rec[-1]_rec[-24];TICK;MPP_Z6*Z10*Z15*Z11;DT(1,1,6)rec[-1]_rec[-24];TICK;MPP_Z11*Z10*Z17*Z12;DT(2,2,7)rec[-1]_rec[-14]_rec[-24];TICK;MPP_Y6*Y10*Y15*Y11*Y17*Y12*Y7;OI(0)rec[-1]_rec[-8]_rec[-9]_rec[-10]_rec[-11]_rec[-12]_rec[-13]_rec[-14]_rec[-16]_rec[-17]_rec[-18]_rec[-19]_rec[-20]_rec[-21]_rec[-22]_rec[-23]_rec[-24]}{here})
from the predecessor of MSC, Itogawa et al.~\cite{Itogawa2025},
and a minimally modified long single-check version
(Crumble link \href{https://algassert.com/crumble\#circuit=Q(0,0)0;Q(0,1)1;Q(0,3)2;Q(0,4)3;Q(1,0)4;Q(1,1)5;Q(1,2)6;Q(1,3)7;Q(1,4)8;Q(2,0)9;Q(2,1)10;Q(2,2)11;Q(2,3)12;Q(2,4)13;Q(3,0)14;Q(3,1)15;Q(3,2)16;Q(3,3)17;Q(3,4)18;MPP_Y1*Y5*Y10*Y6*Y12*Y7*Y2;POLYGON(0,0,1,0.25)2_1_6_7;POLYGON(0,1,0,0.25)5_7_12;POLYGON(1,0,0,0.25)1_6_10;DT(0,1,0)rec[-1];TICK;MPP_X1*X6*X7*X2;TICK;MPP_X1*X5*X10*X6;TICK;MPP_X6*X5*X12*X7;TICK;MPP_Z1*Z6*Z7*Z2;TICK;MPP_Z1*Z5*Z10*Z6;TICK;MPP_Z6*Z5*Z12*Z7;TICK;R_17;RX_16;TICK;CX_16_17;R_18_15;TICK;CX_17_18_16_15;R_14_13_11;TICK;CX_15_14_18_13_16_11;R_8_9;TICK;CX_14_9_13_8;R_4_3;TICK;CX_8_3_9_4;R_0;TICK;CX_4_0;S_1_5_7;S_DAG_2_6_12_10;TICK;CX_8_7_13_12_4_5_9_10_11_6_0_1_3_2;TICK;CX_4_0;S_6_2_12_10;S_DAG_5_1_7;TICK;CX_8_3_9_4;M_0;DT(0,0,1)rec[-1];TICK;CX_14_9_13_8;M_4_3;DT(1,0,2)rec[-2];DT(0,4,2)rec[-1];TICK;CX_15_14_18_13_16_11;M_9_8;DT(2,0,3)rec[-2];DT(1,4,3)rec[-1];TICK;CX_17_18_16_15;M_14_11_13;DT(3,0,4)rec[-3];DT(2,2,4)rec[-2];DT(2,4,4)rec[-1];TICK;CX_16_17;M_15_18;DT(3,1,5)rec[-2];DT(3,4,5)rec[-1];TICK;M_17;MX_16;DT(3,3,6)rec[-2];DT(0,1,6)rec[-1]_rec[-19];TICK;MPP_X1*X6*X7*X2;DT(0,1,7)rec[-1]_rec[-19];TICK;MPP_X1*X5*X10*X6;DT(0,1,8)rec[-1]_rec[-19];TICK;MPP_X6*X5*X12*X7;DT(1,2,9)rec[-1]_rec[-19];TICK;MPP_Z1*Z6*Z7*Z2;DT(0,1,10)rec[-1]_rec[-19];TICK;MPP_Z1*Z5*Z10*Z6;DT(0,1,11)rec[-1]_rec[-19];TICK;MPP_Z6*Z5*Z12*Z7;DT(1,2,12)rec[-1]_rec[-19];TICK;MPP_Y1*Y5*Y10*Y6*Y12*Y7*Y2;OI(0)rec[-1]_rec[-8]}{here}).
Sahay et al.~\cite[\S B.3.c]{Sahay2026} use the latter technique,
exploiting 2D-nonlocal connectivity
to pipeline the (un)encoding of the GHZ states
for its two consecutive long single checks.
Under 2D-local connectivity this pipelining is unavailable,
so it is desirable to keep the double check
especially since it uses nearly half the spacetime volume
per logical measurement compared to the long single check.

We therefore offer an alternative solution,
\cref{fig:d3a6f2,fig:d5a19f13},
which keeps the double check but adds Z flags
to detect all low-fault-count malignant configurations
(see \cref{fig:d3a6f2_conventional_diagram} for
the conventional circuit diagram of \cref{fig:d3a6f2}).
Our Z flags are 2D local in that
	each interacts only with a pair of neighbouring qubits,
but they are not 2D \emph{nearest-neighbour} since
	the existing interactions become next-nearest-neighbour.
We therefore do not intend for them to be physically implemented verbatim,
but rather we treat them as a proof of concept.
In \cref{sec:pre_escape_performance} we show
Monte Carlo simulation results with and without these flags,
and in \cref{sec:flag_design_process}
we briefly comment on our flag design process.

\begin{figure}
	\centering
	\begin{tikzpicture}
	\node[anchor=south west, inner sep=0] (image)
		{\includegraphics[width=0.9\textwidth]{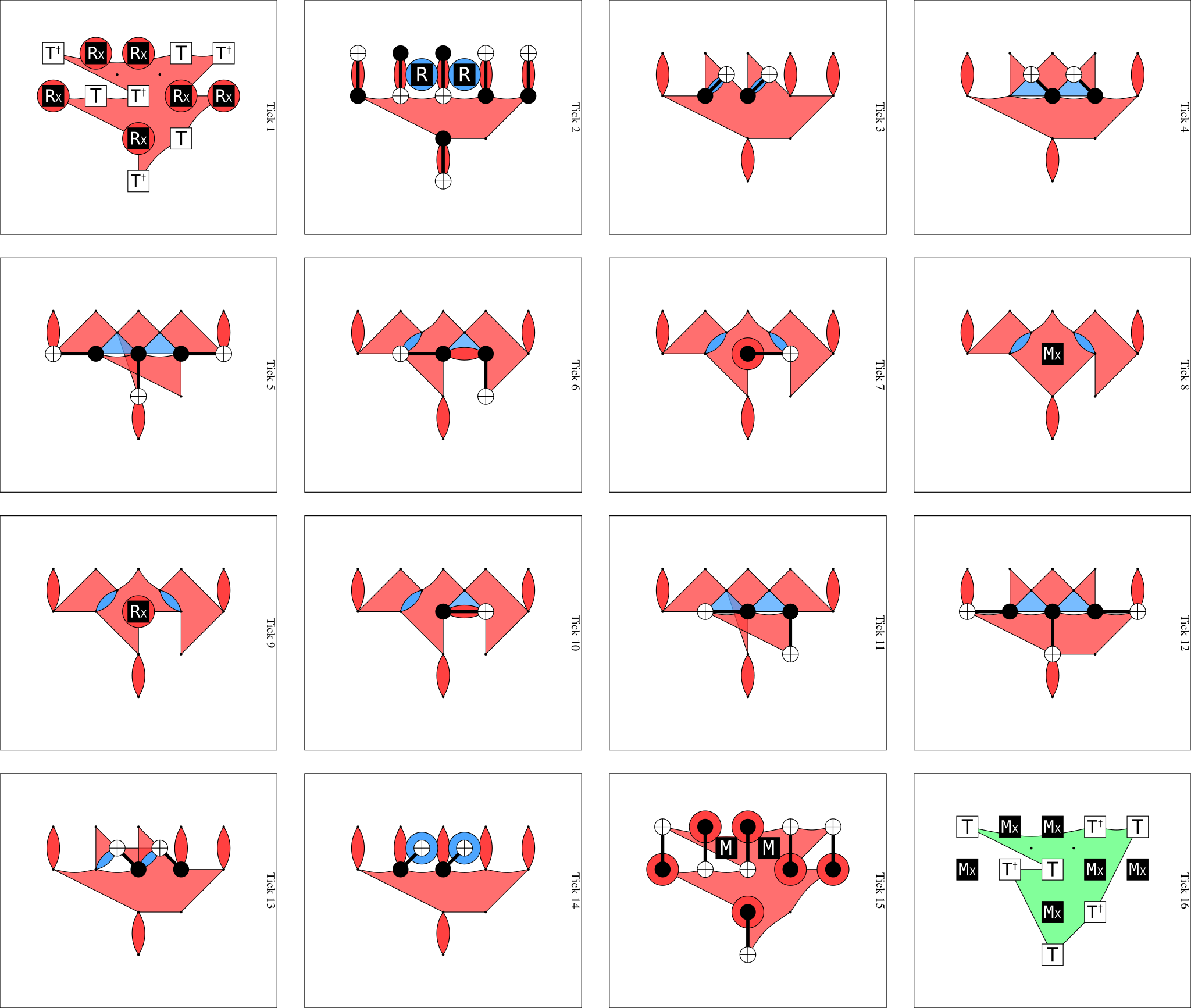}};
	\begin{scope}[x={(image.south east)}, y={(image.north west)}]
		\foreach \x/\y in {0.511/0.766, 0.767/0.766, 0.00/0.001, 0.256/0.001}
		\draw[draw=zblue, line width=1.5pt]
			(\x,\y) rectangle ++(0.234,0.233);
	\end{scope}
	\end{tikzpicture}
	\caption{A Z-flagged version of the distance-3 double-check circuit
	that detects the four fault configurations
	shown in \cref{fig:all_2_fault_configurations}.
	The two added flag qubits are those prepared in $\ket{0}$ in Tick 2,
	and the four additional CX layers are boxed in blue.
	Red, green, and blue shapes represent
	$\X$, $\H_\pm :=\frac{1}{\sqrt{2}}(\X \pm \Y)$, and $\Z$
	\emph{detecting regions} (defined in \cref{sec:reducing_code_conversion_to_codespace_projection}),
	respectively.
	Click \href{https://algassert.com/crumble\#circuit=Q(0,0)0;Q(0,1)1;Q(1,0)2;Q(1,1)3;Q(1.5,0.5)4;Q(2,0)5;Q(2,1)6;Q(2,2)7;Q(2,3)8;Q(2.5,0.5)9;Q(3,0)10;Q(3,1)11;Q(3,2)12;Q(4,0)13;Q(4,1)14;POLYGON(0,0,1,0.25)10_13_12_6;POLYGON(0,1,0,0.25)8_12_6_3;POLYGON(1,0,0,0.25)3_6_10_0;TICK;MPP_Y0*Y10*Y13*Y12*Y8*Y6*Y3;TICK;MPP_X10*X0*X3*X6;TICK;MPP_X3*X6*X12*X8;TICK;MPP_X13*X10*X6*X12;TICK;MPP_Z10*Z0*Z3*Z6;TICK;MPP_Z8*Z12*Z6*Z3;TICK;MPP_Z13*Z10*Z6*Z12;TICK;RX_14_11_5_2_7_1;S_3_10_12;S_DAG_0_6_8_13;TICK;CX_1_0_11_10_7_8_5_6_2_3_14_13;R_4_9;TICK;CX_3_4_6_9;TICK;CX_6_4_11_9;TICK;CX_3_1_6_7_11_14;TICK;CX_6_3_11_12;TICK;CX_6_11;TICK;MX_6;DT(0,0,0)rec[-1]_rec[-8];TICK;RX_6;TICK;CX_6_11;TICK;CX_6_3_11_12;TICK;CX_3_1_6_7_11_14;TICK;CX_6_4_11_9;TICK;CX_3_4_6_9;TICK;CX_1_0_11_10_7_8_14_13_5_6_2_3;M_4_9;DT(1.5,0.5,1)rec[-2];DT(2.5,0.5,1)rec[-1];TICK;MX_14_11_5_2_7_1;S_6_13_8_0;S_DAG_3_10_12;DT(4,1,2)rec[-6];DT(3,1,2)rec[-5];DT(2,1,2)rec[-4]_rec[-9];DT(1,0,2)rec[-3];DT(2,2,2)rec[-2];DT(0,1,2)rec[-1];TICK;MPP_X10*X0*X3*X6;DT(3,0,3)rec[-1]_rec[-10]_rec[-16];TICK;MPP_X8*X12*X6*X3;DT(1,1,4)rec[-1]_rec[-11]_rec[-16];TICK;MPP_X13*X10*X6*X12;DT(4,0,5)rec[-1]_rec[-12]_rec[-16];TICK;MPP_Z10*Z0*Z3*Z6;DT(3,0,6)rec[-1]_rec[-16];TICK;MPP_Z8*Z12*Z6*Z3;DT(2,3,7)rec[-1]_rec[-16];TICK;MPP_Z13*Z10*Z6*Z12;DT(4,0,8)rec[-1]_rec[-16];TICK;MPP_Y0*Y10*Y13*Y12*Y8*Y6*Y3;OI(0)rec[-1]_rec[-8]_rec[-9]_rec[-12]_rec[-13]}{here}
	for the Crumble link.}
	\label{fig:d3a6f2}
\end{figure}

\begin{figure}
	\centering
	\begin{tikzpicture}
	\node[anchor=south west, inner sep=0] (image)
		{\includegraphics[width=0.97\textwidth]{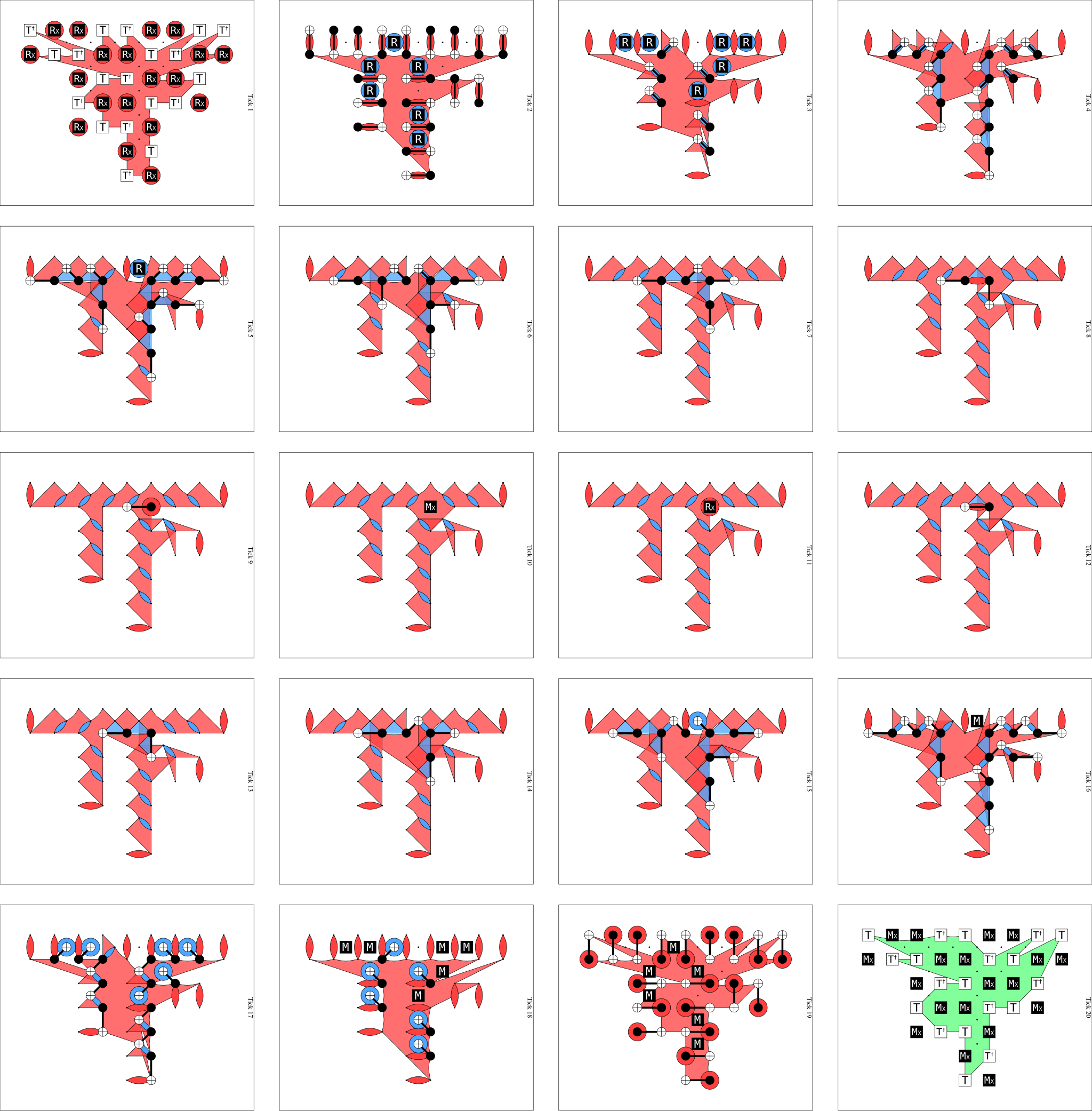}};
	\begin{scope}[x={(image.south east)}, y={(image.north west)}]
		\foreach \x/\y in {0.512/0.814, 0.256/0.001}
		\draw[draw=zblue, line width=1.5pt]
			(\x,\y) rectangle ++(0.234,0.185);
	\end{scope}
	\end{tikzpicture}
	\caption{A Z-flagged version of the distance-5 double-check circuit
	that has fault distance 5.
	The 13 added flag qubits are those prepared in $\ket{0}$,
	and the two additional circuit layers are boxed in blue.
	Click \href{https://algassert.com/crumble\#circuit=Q(0,0)0;Q(0,1)1;Q(1,0)2;Q(1,1)3;Q(1.5,0.5)4;Q(2,0)5;Q(2,1)6;Q(2,2)7;Q(2,3)8;Q(2,4)9;Q(2.5,0.5)10;Q(2.5,1.5)11;Q(2.5,2.5)12;Q(3,0)13;Q(3,1)14;Q(3,2)15;Q(3,3)16;Q(3,4)17;Q(3.5,0.5)18;Q(4,0)19;Q(4,1)20;Q(4,2)21;Q(4,3)22;Q(4,4)23;Q(4,5)24;Q(4,6)25;Q(4.5,0.5)26;Q(4.5,1.5)27;Q(4.5,2.5)28;Q(4.5,3.5)29;Q(4.5,4.5)30;Q(5,0)31;Q(5,1)32;Q(5,2)33;Q(5,3)34;Q(5,4)35;Q(5,5)36;Q(5,6)37;Q(5.5,0.5)38;Q(5.5,1.5)39;Q(6,0)40;Q(6,1)41;Q(6,2)42;Q(6,3)43;Q(6.5,0.5)44;Q(7,0)45;Q(7,1)46;Q(7,2)47;Q(7,3)48;Q(8,0)49;Q(8,1)50;POLYGON(0,0,1,0.25)13_19_32_21_15_6;POLYGON(0,0,1,0.25)41_45_49_47;POLYGON(0,0,1,0.25)34_43_36_23;POLYGON(0,1,0,0.25)8_15_6_3;POLYGON(0,1,0,0.25)32_41_47_43_34_21;POLYGON(0,1,0,0.25)17_23_36_25;POLYGON(1,0,0,0.25)3_6_13_0;POLYGON(1,0,0,0.25)19_45_41_32;POLYGON(1,0,0,0.25)15_21_34_23_17_8;TICK;MPP_X0*X13*X6*X3_X19*X45*X41*X32_X15*X21*X34*X23*X17*X8;TICK;MPP_X32*X41*X47*X43*X34*X21_X3*X6*X15*X8_X17*X23*X36*X25;TICK;MPP_X13*X19*X32*X21*X15*X6_X34*X43*X36*X23_X41*X45*X49*X47;TICK;MPP_Z0*Z13*Z6*Z3_Z19*Z45*Z41*Z32_Z15*Z21*Z34*Z23*Z17*Z8;TICK;MPP_Z32*Z41*Z47*Z43*Z34*Z21_Z3*Z6*Z15*Z8_Z17*Z23*Z36*Z25;TICK;MPP_Z13*Z19*Z32*Z21*Z15*Z6_Z34*Z43*Z36*Z23_Z41*Z45*Z49*Z47;TICK;MPP_Y3*Y6*Y0*Y13*Y19*Y32*Y45*Y41*Y47*Y43*Y34*Y36*Y25*Y23*Y17*Y8*Y15*Y21*Y49;TICK;RX_31_40_2_5_7_9_22_24_48_1_14_20_46_42_16_33_35_37_50;S_3_13_32_15_17_36_34_47_45;S_DAG_19_21_23_25_41_43_49_6_8_0;TICK;CX_1_0_2_3_5_6_14_13_20_19_31_32_40_41_46_45_42_43_48_47_33_21_22_34_35_23_24_36_37_25_7_15_16_8_9_17_50_49;R_11_27_30_29_12_18;TICK;CX_16_12_35_29_36_30_15_11_33_27_14_18;R_39_28_44_38_4_10;TICK;CX_16_17_36_37_14_11_32_27_33_28_35_30_34_29_15_12_42_39_46_44_41_38_3_4_6_10;TICK;CX_35_36_15_16_14_10_42_47_33_39_32_38_34_28_46_50_41_44_3_1_6_4;R_26;TICK;CX_34_35_14_15_33_42_20_18_41_46_32_26_6_3;TICK;CX_14_6_33_34_20_26_32_41;TICK;CX_20_14_32_33;TICK;CX_32_20;TICK;MX_32;DT(1,1,0)rec[-1]_rec[-2];TICK;RX_32;TICK;CX_32_20;TICK;CX_20_14_32_33;TICK;CX_14_6_33_34_20_26_32_41;TICK;CX_34_35_14_15_33_42_20_18_41_46_32_26_6_3;TICK;CX_35_36_15_16_14_10_42_47_33_39_32_38_34_28_46_50_41_44_3_1_6_4;M_26;DT(4.5,0.5,1)rec[-1];TICK;CX_16_17_36_37_14_11_32_27_33_28_35_30_34_29_15_12_42_39_46_44_41_38_3_4_6_10;TICK;CX_16_12_35_29_36_30_15_11_33_27_14_18;M_39_28_44_38_4_10;DT(5.5,1.5,2)rec[-6];DT(4.5,2.5,2)rec[-5];DT(6.5,0.5,2)rec[-4];DT(5.5,0.5,2)rec[-3];DT(1.5,0.5,2)rec[-2];DT(2.5,0.5,2)rec[-1];TICK;CX_1_0_2_3_5_6_14_13_20_19_31_32_40_41_46_45_42_43_48_47_33_21_22_34_35_23_24_36_37_25_16_8_9_17_7_15_50_49;M_11_27_30_29_12_18;DT(2.5,1.5,3)rec[-6];DT(4.5,1.5,3)rec[-5];DT(4.5,4.5,3)rec[-4];DT(4.5,3.5,3)rec[-3];DT(2.5,2.5,3)rec[-2];DT(3.5,0.5,3)rec[-1];TICK;MX_22_24_7_9_2_5_31_40_48_1_14_20_46_42_16_33_35_37_50;S_19_21_23_25_41_43_6_8_0_49;S_DAG_3_13_32_15_17_36_34_47_45;DT(4,3,4)rec[-19];DT(4,5,4)rec[-18];DT(2,2,4)rec[-17];DT(2,4,4)rec[-16];DT(1,0,4)rec[-15];DT(2,0,4)rec[-14];DT(5,1,4)rec[-13]_rec[-33];DT(6,0,4)rec[-12];DT(7,3,4)rec[-11];DT(0,1,4)rec[-10];DT(3,1,4)rec[-9];DT(4,1,4)rec[-8];DT(7,1,4)rec[-7];DT(6,2,4)rec[-6];DT(3,3,4)rec[-5];DT(5,2,4)rec[-4];DT(5,4,4)rec[-3];DT(5,6,4)rec[-2];DT(8,1,4)rec[-1];TICK;MPP_X0*X13*X6*X3_X19*X45*X41*X32_X15*X21*X34*X23*X17*X8;DT(0,0,5)rec[-3]_rec[-55];DT(4,0,5)rec[-2]_rec[-36]_rec[-54];DT(3,2,5)rec[-1]_rec[-53];TICK;MPP_X32*X41*X47*X43*X34*X21_X3*X6*X15*X8_X17*X23*X36*X25;DT(5,1,6)rec[-3]_rec[-39]_rec[-55];DT(1,1,6)rec[-2]_rec[-54];DT(3,4,6)rec[-1]_rec[-53];TICK;MPP_X13*X19*X32*X21*X15*X6_X34*X43*X36*X23_X45*X49*X47*X41;DT(3,0,7)rec[-3]_rec[-42]_rec[-55];DT(5,3,7)rec[-2]_rec[-54];DT(6,1,7)rec[-1]_rec[-53];TICK;MPP_Z0*Z13*Z6*Z3_Z19*Z45*Z41*Z32_Z15*Z21*Z34*Z23*Z17*Z8;DT(0,0,8)rec[-3]_rec[-55];DT(4,0,8)rec[-2]_rec[-54];DT(3,2,8)rec[-1]_rec[-53];TICK;MPP_Z32*Z41*Z47*Z43*Z34*Z21_Z3*Z6*Z15*Z8_Z17*Z23*Z36*Z25;DT(5,1,9)rec[-3]_rec[-55];DT(1,1,9)rec[-2]_rec[-54]}{here}
	for the Crumble link.}
	\label{fig:d5a19f13}
\end{figure}
\begin{SCfigure}
	\includegraphics[width=0.65\linewidth]{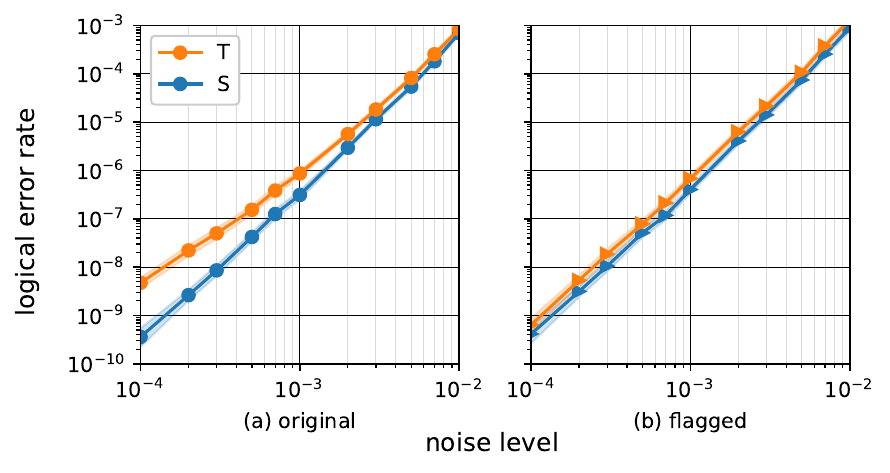}
	\caption{The logical error rate per kept shot of
	distance-3 cultivation before the escape stage.
	Blue (orange) data correspond to cultivating S (T) states.
	Shading shows the region within a factor of 1000 of the maximum likelihood.}
	\label{fig:d3_pre_escape_ler}
\end{SCfigure}

\begin{SCfigure}
	\includegraphics[width=0.47\linewidth]{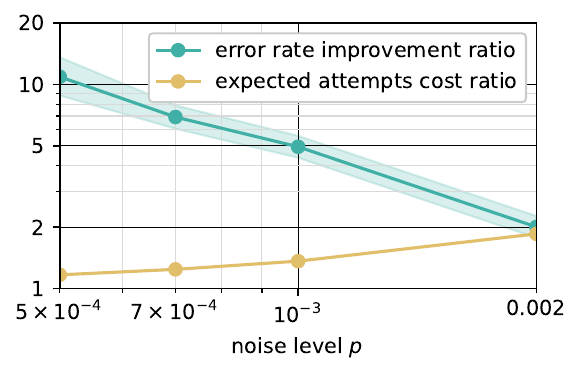}
	\caption{The pre-escape performance of distance-5 T-state cultivation,
	with and without Z flags.
	Turquoise:
	the logical error rate ratio (original / flagged);
	sand:
	the expected attempts ratio (flagged / original);
	shading shows 95\% Bayesian credible intervals.
	This figure and \cref{fig:d5_pre_escape_ler} are derived from the same simulation data.}
	\label{fig:improvement_and_cost_ratios}
\end{SCfigure}

\subsection{Pre-Escape Performance}\label{sec:pre_escape_performance}
To see how our Z-flagged MSC performs in practice,
we use the fast universal simulator SymFT~\cite{Fang2026a}
to simulate original and Z-flagged MSC
up to and excluding the escape stage.
\Cref{fig:d3_pre_escape_ler} shows that our Z flags recover
the intended fault distance of MSC-3,
thus pinning the final double check as
the sole culprit degrading its fault distance.
For MSC-3 in practice we can see that these Z flags are overkill,
since the malignant 2-fault configurations they detect
barely outweigh the higher-fault-count configurations they introduce.

For MSC-5 however \cref{fig:d5_pre_escape_ler,fig:improvement_and_cost_ratios} show they are practical:
in addition to the numbers stated in the abstract,
at noise level $p =\num{5e-4}$
the Z flags improve the pre-escape logical error rate by 10.9$\times$,
costing a 1.17$\times$ increase in attempts per kept shot.
Note the simulated circuits have an intended fault distance of 4
as \cref{sec:magic_state_cultivation} explains,
so the true fault distances are 3 without and 4 with flags.
However,
a weighted-least-squares fit gives log--log slopes of
3.8 without and 5.1 with flags,
indicating the simulated noise levels do not yet reach the asymptotic scaling regime.
Our accompanying code~\cite{Chan2026c_quantum_bibstyle}
also includes an unsimulated circuit for Z-flagged MSC-5 with an intended fault distance of 5
i.e.\ with four instead of three stabiliser rounds before its distance-5 double check.

The analysis from Wan and Zapirain~\cite[Appendix~H.3--4]{Wan2026a}
more rigorously verifies that with our Z flags:
MSC-3,
three-round MSC-5,
and four-round MSC-5
indeed have fault distances 3, 4, and 5, respectively.
In the third case,
our Z flags restore the pre-escape logical error rate
from $\upTheta(p^3)$ to $\upTheta(p^5)$.
Pre-escape logical error rates are relevant to full MSC
if the escape-stage decoder confidence threshold is set high:
Chase and Labib~\cite[\S 4.2.2]{Chase2026} showed that in this regime,
the T--S discrepancy persists after the escape stage,
thus most of the end-to-end logical error
comes from the injection and cultivation stages
rather than the purely Clifford escape stage.

\subsection{Flag Design Process}\label{sec:flag_design_process}
For distance 3:
we added the minimum number of Z flags,
preferring short-range interactions,
that together detects each X fault in \cref{fig:all_2_fault_configurations}.
For distance 5:
there are too many ($3 +601$) malignant configurations of \num{<5} faults
to individually design flags for.
Instead,
we continued the trend from distance 3,
adding a Z flag between each of the 18 consecutive qubit pairs
in the tree along which the GHZ state is (un)encoded;
we verified that the resulting double check has fault distance 5.
We then removed the flag at each of the 5 leaves of the tree
as we found this did not reduce fault distance.
Finally,
we minimised the lifetime, hence the idling error, of each flag qubit.
\Cref{alg:probability_trivial_syndrome,prop:resulting_fidelity}
were crucial as they enabled the fast verification of whether
a candidate Z-flagged double check had the intended fault distance.
\section{Conclusion}\label{sec:conclusion}
Our paper makes three main contributions.
First,
we analytically show that distance-3 (-5) T-state cultivation
has fault distance 2 (3) before its escape stage;
this is the main reason for the T--S discrepancy
in logical error rate.
Second,
we derive a general formula for,
and an algorithm for calculating,
the acceptance probability of
an arbitrary encoded mixed state suffering a physical Clifford error.
Third,
we design 2D-local flags that restore
the degraded fault distance of T-state cultivation,
and which are practical for distance 5.
These flags demonstrate the utility of considering
coherent non-Pauli errors arising from non-Clifford circuits
for improving their fault tolerance.
There are various avenues for future work.

Surti et al.~\cite{Surti2026} show that
in many logical magic state preparation circuits
(even those with multiple layers of non-Cliffords),
each Pauli error propagates to a Clifford error.
\Cref{prop:resulting_fidelity} could be generalised to
the same level of generality as \cref{theorem:probability_trivial_syndrome}
using the same proof technique.
These two facts could then be combined with
\cref{alg:probability_trivial_syndrome}
to quickly and analytically verify the fault distance
(by enumerating low-fault-count fault configurations)
of other circuits, such as
the double checks~\cite{Chen2026,Vaknin2025,Claes2025}
or the constant-depth logical measurement by gauging~\cite{Hetenyi2026}
from more recent cultivation proposals.
The transversal-T code-switching protocol
from Daguerre et al.~\cite{Daguerre2025,Daguerre2025a},
which is an alternative to cultivation,
could also be analysed in this way.

One motivation for designing the Z-flagged circuits
was to retain the use of the double check
for its low spacetime cost under 2D connectivity constraints.
As \cref{sec:z_flagged_double_checks} mentions however,
the Z flags themselves complicate the connectivity
to next-nearest-neighbour;
it would be interesting to devise a qubit layout
that maintains nearest-neighbour connectivity.
This should be combined with
a more rigorous and systematic flag design process
than the hand-optimised one described in \cref{sec:flag_design_process}.
Additionally,
it is unknown how degraded the fault distance is for higher-distance MSC;
if the trend continues from MSC-3 and -5,
our Z flags could become increasingly effective at improving
the logical error rate.
\section{Author Contributions}
AS and ZS performed early exploration of the cultivation protocol.
TC and ZC discovered the degraded fault distance.
TC wrote the paper,
constructed its mathematical statements,
designed the Z-flagged double checks,
developed the code,
and ran the numerics.
AS and ZC supervised the project.
All authors revised the paper.

\begin{acknowledgments}
We thank
Craig Gidney,
Theerapat Tansuwannont,
Yutaka Hirano,
Tom Scruby,
Keisuke Fujii,
and Kwok Ho Wan,
for useful discussions.
TC acknowledges the use of Crumble
in designing the flag circuits.
TC acknowledges support from
an EPSRC DTP studentship and
JST ASPIRE Japan Grant Number JPMJAP2319.
AS thanks UKRI for the Future Leaders Fellowship Theory to Enable Practical Quantum Advantage (MR/Y015843/1). ZC acknowledges support from the EPSRC Quantum Technologies Career Acceleration Fellowship (UKRI1226).
We acknowledge
the use of
the University of Oxford Advanced Research Computing
(ARC)
facility~\cite{Richards2015_quantum_bibstyle},
and three EPSRC projects:
QCS Hub (EP/T001062/1),
RoaRQ (EP/W032635/1),
and SEEQA (EP/Y004655/1). 
We acknowledge the use of ChatGPT Edu
to assist with the programming and testing
of our emulation code,
and we take full responsibility for its content.
For the purpose of Open Access, the authors have applied a CC BY public copyright licence to any Author Accepted Manuscript version arising from this submission.
\end{acknowledgments}

\paragraph{Note added:}
Two related preprints were announced in the same arXiv posting as this preprint.
Wan and Zapirain~\cite{Wan2026a} exactly compute the logical error rate of
cultivation using Pauli propagation~\cite{Rudolph2026}
and tensor-network constructions.
Their approach is similar to that in \cref{sec:alternative_clifford_error_analysis}
but further quotients the terms in each Pauli superposition by the stabilisers,
helping keep the calculation tractable.
Before announcement,
we privately communicated to them our finding that
the fault distances of MSC-3 and -5 are respectively 2 and 3,
which their analysis now corroborates.
Hartweg and Piñeiro Orioli~\cite{Hartweg2026} use tensor-network simulation methods
to analyse the T--S discrepancy and propose cultivation protocols
hosted entirely on the rotated surface code.

\bibliographystyle{quantum}
\bibliography{tchbib}

\appendix

\section{Formal Definitions}\label{sec:formal_definitions}

\subsection{Codespace Projector}
Given a commuting set $\mathcal G$ of Paulis,
define
\begin{equation}
\upPi_{\mathcal G} :=\prod_{P \in \mathcal G} \frac{\I +P}2
\label{eq:pj_defn}
\end{equation}
as the projector onto their simultaneous +1 eigenspace.
Specifically,
if $\mathcal G$ generates a stabiliser code,
$\upPi_{\mathcal G}$ projects onto its codespace.
In most cases in these appendices,
the representative of a logical operator
appears multiplied by $\upPi_{\mathcal G}$
so the choice of representative makes no difference;
when it is not, we specify the representative.

\subsection{Faults}\label{sec:faults}
We adapt the definitions in Beverland et al.~\cite[\S 3]{Beverland2024}
and more formally define some concepts introduced in \cref{sec:fault_distance}.
Given a stabiliser circuit,
define a set $\mathcal{F}$ of independent events called \emph{faults}.
Each fault $f$ specifies a pair $(P_f, \v\sigma_f)$,
where the \emph{resultant effect} $P_f$ is a Pauli acting on the data qubits,
and the \emph{signature} $\v\sigma_f$ is a binary vector
whose support indicates the detectors the fault flips.
Each fault induces a function $\pi_f(p)$
that outputs its occurrence probability,
and assume $\pi_f(p) =\upTheta(p)$ for noise level $p \ll 1$.

To construct $\mathcal{F}$ from a set of independent error events,
propagate each one to the end of the circuit,
thereby associating it with a pair $(P, \v \sigma)$.
For each distinct pair,
define one fault representing all its associated error events.
Its occurrence probability $\pi_f(p)$ is
the probability an odd number of those events occur.
Coarse-graining from error events to faults
makes the error analysis of circuits much faster.

A \emph{$|\mathcal C|$-fault configuration} (or \emph{configuration} for short)
is a combination $\mathcal C \subseteq \mathcal F$ of faults whose
resultant effect, signature, and occurrence probability is given respectively by
\begin{align}
	P(\mathcal C) &=\prod_{f \in \mathcal C} P_f \\
	\v\sigma(\mathcal C) &=\bigoplus_{f \in \mathcal C} \v\sigma_f \\
	\pi_\mathcal C(p)
	&=`\Big[\prod_{f \in \mathcal C} \frac{\pi_f(p)}{1 -\pi_f(p)}]
	\prod_{f \in \mathcal F} [1 -\pi_f(p)]
	\label{eq:configuration_probability}
\end{align}
where the second product in \cref{eq:configuration_probability}
is independent of $\mathcal C$.
\section{Additional Figures}\label{sec:additional_figures}
In \cref{sec:thorough_analysis} we enumerate all undetected fault configurations
in the distance-3 and -5 double checks.
\Cref{fig:fcc_against_fault_count} shows
the number of such configurations,
after grouping by fault count and whether they are benign or malignant.
\begin{figure}[t]
	\centering
	\includegraphics[width=0.72\linewidth]{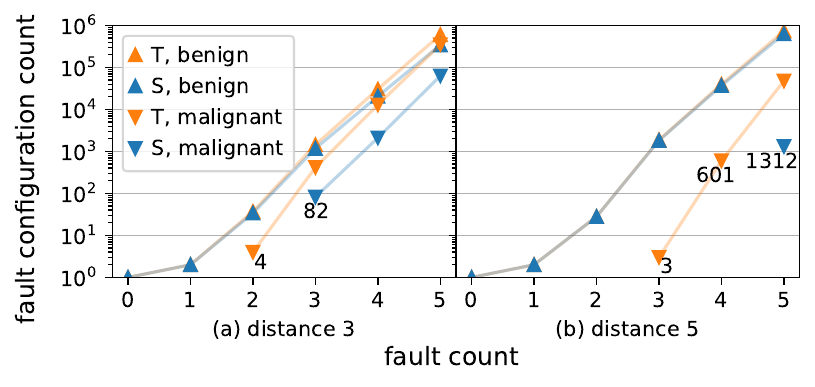}
	\caption{The number of undetected fault configurations
	in the double-check circuit for magic state cultivation,
	grouped by fault count i.e.\ number of faults in the configuration.
	(a) and (b) correspond to distance-3 and -5 cultivation respectively.
	In both subfigures:
	blue (orange) data correspond to cultivating S (T) states.
	Benign (malignant) configurations
	lead to a logical fidelity of 1 (0).
	Five low-fault-count malignant counts are annotated.}
	\label{fig:fcc_against_fault_count}
\end{figure}

In \cref{sec:z_flagged_double_checks} we show
our Z-flagged distance-3 double check as a detector slice diagram.
\Cref{fig:d3a6f2_conventional_diagram} shows
the same circuit more conventionally.
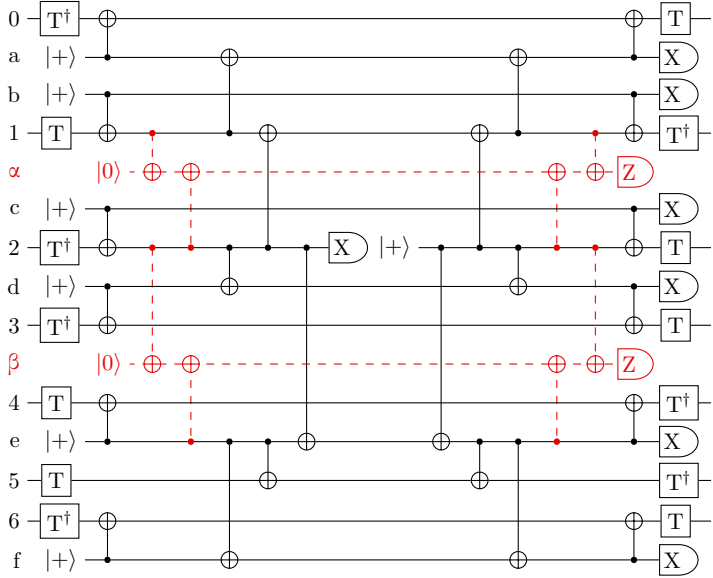
\begin{SCfigure}
	\begin{tikzpicture}
\begin{yquant}

qubit {0} q0;
qubit {a} qa;
qubit {b} qb;
qubit {1} q1;
[darkred]
qubit {\alpha} qalpha;
qubit {c} qc;
qubit {2} q2;
qubit {d} qd;
qubit {3} q3;
[darkred]
qubit {\beta} qbeta;
qubit {4} q4;
qubit {e} qe;
qubit {5} q5;
qubit {6} q6;
qubit {f} qf;

setstyle {darkred, dashed} qalpha,qbeta;
discard qa,qb,qalpha,qbeta,qc,qd,qe,qf;
hspace {1mm} -;
box {T} q1,q4,q5;
box {T$^\dag$} q0,q2,q3,q6;
init {$\ket{+}$} qa,qb,qc,qd,qe,qf;
align -;

cnot q0|qa;
cnot q4|qe;
cnot q3|qd;
cnot q2|qc;
cnot q1|qb;
cnot q6|qf;
[darkred]
init {$\ket{0}$} qalpha,qbeta;

[dashedcnot]
cnot qalpha|q1;
[dashedcnot]
cnot qbeta|q2;
[dashedcnot]
cnot qalpha|q2;
[dashedcnot]
cnot qbeta|qe;

align -;
cnot qa|q1;
cnot qd|q2;
cnot qf|qe;

cnot q1|q2;
cnot q5|qe;

cnot qe|q2;

dmeter {X} q2;
discard q2;
init {$\ket{+}$} q2;
align q2;

cnot qe|q2;

align -;
cnot q1|q2;
cnot q5|qe;

cnot qa|q1;
cnot qd|q2;
cnot qf|qe;

[dashedcnot]
cnot qalpha|q2;
[dashedcnot]
cnot qbeta|qe;
[dashedcnot]
cnot qalpha|q1;
[dashedcnot]
cnot qbeta|q2;

align -;
cnot q0|qa;
cnot q4|qe;
cnot q3|qd;
cnot q2|qc;
cnot q1|qb;
cnot q6|qf;
[darkred]
dmeter {Z} qalpha,qbeta;
discard qalpha,qbeta;

align -;
dmeter {X} qa,qb,qc,qd,qe,qf;
discard qa,qb,qc,qd,qe,qf;
box {T$^\dag$} q1,q4,q5;
box {T} q0,q2,q3,q6;

hspace {1mm} -;

\end{yquant}
\end{tikzpicture}
	\caption{A Z-flagged version of the distance-3 double-check circuit
	that detects the four fault configurations
	shown in \cref{fig:all_2_fault_configurations}.
	The two added flag qubits {\alpha} and {\beta},
	and the four additional CX layers,
	are drawn red and dashed.}
	\label{fig:d3a6f2_conventional_diagram}
\end{SCfigure}
\section{Reducing Code Conversion to Codespace Projection}
\label{sec:reducing_code_conversion_to_codespace_projection}
We first give two definitions.
\emph{Code conversion} is a quantum circuit,
which may include measurement and feedforward,
that changes encoding and preserves every logical state.
This is more general than code switching
as it may involve conversion into the same code family
e.g.\ surface code to surface code.
Though MSC uses a code conversion technique called grafting in its escape stage,
later MSC variants demonstrate that
other code conversion techniques can be used:
	measurement-based growth is used in~\cite{Sahay2026,Claes2025,Chen2026},
	teleportation via lattice surgery is used in~\cite{Itogawa2025,Hirano2025a,Thorvaldson2026},
	and local unitary transformations are used in~\cite{Itogawa2025,Sahay2026}.

Associated to each detector is a \emph{detecting region} which
(for stabiliser circuits)
is a set of circuit spacetime locations
each labelled by X, Y, or Z~\cite{McEwen2023},
and shows how the Pauli measured by that detector
evolves through the circuit at each timeslice.

\begin{lem}
\label{lem:insert_projector}
Let $V$ be a code conversion circuit
from an $\llbracket n,k\rrbracket$ input code generated by $\mathcal G$.
Suppose $n -k$ detecting regions have
independent timeslices $P_j$ at the input of $V$,
with eigenvalue $+1$ when the corresponding detector is unflipped.
Then in the absence of noise,
for an arbitrary mixed state $\rho$ entering $V$,
the probability those detectors
and any others within or after $V$ are unflipped is
\begin{equation}
	\pr(\v s =\v 0)
	=\tr(\upPi_{\mathcal G} \rho).
	\label{eq:insert_projector}
\end{equation}
Conditioned on this event,
inserting the codespace projector $\upPi_{\mathcal G}$
immediately before $V$ leaves the output state unchanged.
\end{lem}
\begin{proof}
In the absence of noise,
each detector within or after $V$ can be pulled back along its detecting region
to its timeslice $Q$ at the input of $V$.
Code conversion preserves every logical state,
so $Q \in `<\mathcal G>$.
The $n -k =\rk `<\mathcal G>$ timeslices $P_j$ in particular are independent,
so $`<P_j> =`<\mathcal G>$.
Conditioning on their corresponding $n -k$ detectors being unflipped
projects $\rho$ onto the simultaneous $+1$ eigenspace of $\mathcal G$,
making redundant the condition any other detector is unflipped.
This proves \cref{eq:insert_projector}.

Conditioned on this event:
the projection already occurs within or after $V$,
so inserting it beforehand leaves the output unchanged.
\end{proof}
\section{Clifford Error Analysis Proofs}
\label{sec:clifford_error_analysis_proofs}
In this section we prove \cref{theorem:probability_trivial_syndrome,prop:resulting_fidelity}.
First note the following.
We can write the density operator of
the logical state corresponding to \cref{eq:pauli_decomposition} in terms of
	$\mathcal G$ (defining the code)
	and the coefficients $`{\alpha_L}_L$ (defining the logical state):
\begin{align}
\overline{\rho}
=\upPi_{\mathcal G} 2^{-k} \sum_{L \in \mathcal{P}_k} \alpha_L \overline{L}.
\label{eq:density_logical_psi}
\end{align}
Here, $\upPi_{\mathcal G}$ projects onto the codespace of $\mathcal G$,
and the other term projects onto the specific logical state within the codespace.
Note $\upPi_{\mathcal G}$ commutes with each $\overline{L}$.

\subsection{Proof of \texorpdfstring
{\cref{theorem:probability_trivial_syndrome}}
{Theorem~\ref{theorem:probability_trivial_syndrome}}}
Let
$\mathcal A:=\langle\mathcal G\rangle$,
$\mathcal B:=\langle E^\dagger\mathcal G E\rangle$,
$\mathcal K:=\mathcal A\cap\mathcal B$,
and $r :=n-k$.
The probability to observe a trivial syndrome
is the expectation value of $\upPi_{\mathcal G}$
[recall \cref{eq:pj_defn}]:
\begin{align}
	\pr(\v s =\v 0)
	&=\tr`*(\upPi_{\mathcal G} E\overline{\rho} E^\dag) \notag \\
	&=\tr`*(\upPi_{E^\dag \mathcal G E} \overline{\rho}) \notag \\
	&\overset{\eqref{eq:density_logical_psi}}=
	2^{-k}\sum_{L\in\mathcal P_k}\alpha_L \tr\left(
	\upPi_{E^\dagger \mathcal G E}\upPi_{\mathcal G} \overline L
	\right).
	\label{eq:sum_of_traces_fully_expanded}
\end{align}
If $\mathcal B\cap(-\mathcal A)\ne\varnothing$,
then $\upPi_{E^\dagger \mathcal G E}\upPi_{\mathcal G}$
contains $(\I -P)(\I +P) =0$ for some Pauli $P \in \mathcal A$,
so the probability is zero.
Assume otherwise.

Multiply out the two stabiliser projectors:
\begin{equation}
\label{eq:trace_triple}
\pr(\v s =\v 0)
=2^{-2r -k} \sum_{L \in \mathcal P_k} \alpha_L
\sum_{a \in \mathcal A} \sum_{b \in \mathcal B}
\tr`\big(ba \overline{L}).
\end{equation}
Since nontrivial Paulis are traceless,
we need only consider products
$ba \overline{L}$
that are proportional to $\I$.
Thus, each $L$ term contributes iff \cref{eq:logical_membership} holds;
equivalently,
iff $\mathcal B\cap\omega_L\overline L\mathcal A$ is nonempty
for some sign $\omega_L\in\{\pm1\}$.
(This sign is real:
	$\overline{L}$ and all Paulis in $\mathcal A$ and $\mathcal B$ are Hermitian,
	and $\overline{L}$ commutes with each Pauli in $\mathcal A$.
The sign is also unique:
	if $b=\omega\overline L a$ and
	$b'=-\omega\overline L a'$ with
	$b,b'\in\mathcal B$ and $a,a'\in\mathcal A$,
	then $bb'=-aa'\in\mathcal B\cap(-\mathcal A)$,
	a contradiction.)
If the intersection is nonempty and contains $b_L$,
then it is a coset of $\mathcal K$:
\begin{equation}
\mathcal B\cap\omega_L\overline L\mathcal A=b_L\mathcal K,
\end{equation}
so it has $2^{\rk\mathcal K}$ elements.
Each surviving trace contributes
$\tr \omega_L \I= \omega_L 2^n =\omega_L 2^{r +k}$,
and therefore the total prefactor is
$2^{-r+\rk\mathcal K}$,
where
$r-\rk\mathcal K
=\rk\langle\mathcal G\cup E^\dagger\mathcal G E\rangle-r
=s$. \qed

\subsection{Proof of \texorpdfstring
{\cref{prop:resulting_fidelity}}
{Proposition~\ref{prop:resulting_fidelity}}}
The fidelity of the resulting state $\ket{\Omega}$
to $\ket{\overline{T}}$ is $|`<\overline{T}|\Omega>|^2$.
Using the notation of \cref{eq:pj_defn}:
\begin{align}
	\ket{\Omega}
	&=\frac1{\sqrt{\pr(\v s =\v 0)}}
	\upPi_{\mathcal G}
	E\ket{\overline{T}} \notag \\
	|`<\overline{T}|\Omega>|^2 \pr(\v s =\v 0)
	&=\tr`\big(E^\dag \proj{\overline{T}}E \proj{\overline{T}}).
\end{align}
Similar to \cref{eq:density_logical_psi},
we can write $\proj{\overline{T}} =\upPi_{\mathcal G} \frac12(\I +\overline{H_+})$,
where $\overline{H_+}$ here refers to
the logical representative from \cref{eq:logical_H_is_transversal}.
Let $a=0$ if $E$ commutes with $\overline{H_+}$
and $a=1$ if it anticommutes,
then $E^\dag \overline{H_+}E =(-1)^a \overline{H_+}$.
Substitute these two facts in:
\begin{align}
|`<\overline{T}|\Omega>|^2 \pr(\v s =\v 0)
&=\tr`\bigg[E^\dag \upPi_{\mathcal G} E \frac{\I +(-1)^a \overline{H_+}}2 \upPi_{\mathcal G} \frac{\I +\overline{H_+}}2]
\notag \\
&=
\tr`\bigg[\upPi_{E^\dagger \mathcal G E}\upPi_{\mathcal G}
\frac{\I +(-1)^a \overline{H_+}}2
\frac{\I +\overline{H_+}}2].
\end{align}
If $a =1$,
the last two projectors multiply to zero,
so $|`<\overline{T}|\Omega>|^2 =0$.
If $a =0$,
they are the same projector;
since projectors are idempotent,
we keep only one of them in our expression.
That remaining logical projector is multiplied by $\upPi_{\mathcal G}$
i.e.\ acts only on the codespace,
so we can replace it with
the logically equivalent projector
$\sum_{L \in \mathcal P_1} \alpha_L \overline{L}$
from \cref{eq:density_logical_psi}
(here $\alpha_\I =1$ and $\alpha_X =\alpha_Y =\frac1{\sqrt{2}}$).
Then:
\begin{equation}
|`<\overline{T}|\Omega>|^2 \pr(\v s =\v 0)
=2^{-1}\sum_{L\in\mathcal P_k}\alpha_L \tr\left(
\upPi_{E^\dagger \mathcal G E}\upPi_{\mathcal G} \overline L
\right),
\label{eq:sum_of_traces_T_fully_expanded}
\end{equation}
This is identical to the right-hand side of
\cref{eq:sum_of_traces_fully_expanded},
so $|`<\overline{T}|\Omega>|^2 =1$.

The stated time complexity comes from determining $a$;
this requires checking the commutation of each pair of tensor factors,
which takes $\mathcal O(n)$ time. \qed
\section{Acceptance Probability Algorithm}
\label{sec:acceptance_probability_algorithm}
In this section,
$\mathbb{F}_2$ denotes the finite field with $2$ elements.
For $\v v \in \mathbb{F}_2^n$,
write $P_{\v v} :=\bigotimes_{q =1}^n P^{v_q}$ for $P \in\{\X, \Z\}$.

\begin{prop}
\label{prop:probability_runtime}
\Cref{alg:probability_trivial_syndrome} computes
\cref{eq:probability_trivial_syndrome}
in $\mathcal O(n^3+N_\alpha k^2)$ time,
where $N_\alpha \le 4^k$ is the number of nonzero coefficients
$\alpha_L$ defining $\rho$ by \cref{eq:pauli_decomposition}.
\end{prop}

\begin{algorithm}[H]
\caption{Probability of a trivial syndrome after suffering a Clifford error.}
\label{alg:probability_trivial_syndrome}
\begin{algorithmic}[1]
\Require Encoder $C$, Clifford error $E$,
and logical coefficients $\{\alpha_L:L\in\mathcal P_k\}$.
\Ensure $\pr(\v s=\v 0)$.
\State Set $F\gets C^\dag EC$ and build $R$ as in \cref{eq:forget_check_Z_matrix}.
\State Row-reduce $R$ to obtain $\rk R$, a basis $\mathcal{B}$ for $\ker R$,
an image test, and a solver for $R\v x=\v y$.
\State Precompute the data to evaluate $\chi$ using
\cref{eq:chi_phase_formula}.
\ForAll{$\v e \in \mathcal{B}$}
	\If{$\chi(\v e)=1$}
		\State \Return $0$
	\EndIf
\EndFor
\State $\eta\gets0$
\ForAll{$L\in\mathcal P_k$ with $\alpha_L\ne0$}
	\If{$\v \rho(\I^{\otimes r}\otimes L) \in\im R$}
	\label{line:image_membership_test}
		\State $\eta\gets\eta+(-1)^{\chi(\v u_L)}\alpha_L$,
		where $\v u_L$ is any solution to
		$R\v u_L =\v \rho(\I^{\otimes r} \otimes L)$
	\EndIf
\EndFor
\State \Return $2^{-\rk R}\eta$
\end{algorithmic}
\end{algorithm}

\begin{proof}
Let $r :=n-k$.
It helps to work in unencoded space:
choose an encoder $C$ such that the encoded state is
$\overline{\rho} =C \big(\proj{0}^{\otimes r} \otimes \rho\big) C^\dag$.
Then the unencoded stabiliser generators are simply
$\mathcal Z:=\{\Z_1,\ldots,\Z_r\}$
and the unencoded Clifford error is $F:=C^\dag EC$.
Write the transformed stabiliser element corresponding to
$\v u\in\mathbb{F}_2^r$ as $\Z'_{\v u}:=F^\dag\Z_{\v u}F$.
In unencoded space, \cref{eq:trace_triple} becomes
\begin{equation}
\pr(\v s=\v 0)
=
2^{-2r-k}
\sum_{L\in\mathcal P_k}\alpha_L
\sum_{\v u,\v v\in\mathbb{F}_2^r}
\tr\left[\Z_{\v u}'\Z_{\v v}(\I^{\otimes r}\otimes L)\right].
\label{eq:single_conjugated_trace_expansion}
\end{equation}
As before,
we need only consider products
$\Z_{\v u}'\Z_{\v v}(\I^{\otimes r}\otimes L)$
that are proportional to $\I$.
For a given $L$,
this occurs if there exists a $\Z_{\v u}'$ whose...
\begin{enumerate}
	\item first $r$ tensor factors are in $\{\I, \Z\}$,
	\item last $k$ tensor factors match $L$ up to a sign.
\end{enumerate}
These two conditions can be restated as
the single equation
$R\v u=\v \rho(\I^{\otimes r}\otimes L)$,
where we have defined
$\v \rho(P)
:=(\v x^\T\mid z_{r+1},\ldots,z_n)^\T
\in\mathbb{F}_2^{n+k}$
for a Hermitian $n$-qubit Pauli $P$
whose binary symplectic vector is
$(\v x^\T \mid\v z^\T)^\T$,
and the matrix
\begin{equation}
R:=\big[\v \rho(F^\dag\Z_1F),\ldots,\v \rho(F^\dag\Z_rF)\big].
\label{eq:forget_check_Z_matrix}
\end{equation}
There can be multiple solutions to
$R\v u=\v \rho(\I^{\otimes r}\otimes L)$;
we define the \emph{fibre} over $L$ as the solution set:
\begin{equation}
\Phi_L:=
\{\v u\in\mathbb{F}_2^r:R\v u=\v \rho(\I^{\otimes r}\otimes L)\}.
\label{eq:fibre_defn}
\end{equation}
For $\v u\in\Phi_L$,
there is a unique $\v v(\v u)\in\mathbb{F}_2^r$ such that
$\Z_{\v v(\v u)}$ matches the first $r$ tensor factors of $\Z_{\v u}'$
i.e.
\begin{equation}
\Z'_{\v u}\Z_{\v v(\v u)}(\I^{\otimes r}\otimes L)
=(-1)^{\chi(\v u)}\I.
\label{eq:chi_defn}
\end{equation}
\Cref{lem:trivial_syndrome_phase_evaluation}
gives an explicit formula for $\chi: \mathbb{F}_2^r \to \mathbb{F}_2$
and shows that it depends only on $\v u$, not on $L$.

The zero case of \cref{theorem:probability_trivial_syndrome}
is detected on $\ker R$:
by \cref{eq:chi_defn} with $L=\I$,
$\chi(\v k)=1$ iff
$\Z'_{\v k}\in-\langle\mathcal Z\rangle$.
By \cref{lem:kernel_chi_linear},
checking $\chi$ on a basis of $\ker R$
detects precisely whether this happens.

Assume the algorithm does not return zero.
\Cref{line:image_membership_test} selects
exactly the $L$'s satisfying \cref{eq:logical_membership}.
For any solution $\v u_L\in\Phi_L$,
\cref{eq:chi_defn} identifies the sign in
\cref{eq:probability_trivial_syndrome} as
$\omega_L=(-1)^{\chi(\v u_L)}$.
This sign is solution-independent by
\cref{theorem:probability_trivial_syndrome}.
Finally,
$\rk R=s$
because $R$ records the transformed stabilisers modulo
the original unencoded stabiliser group.
Thus the accumulator $\eta$ is the sum in
\cref{eq:probability_trivial_syndrome},
and the algorithm returns the prefactor $2^{-s}=2^{-\rk R}$.

It remains to prove the stated time complexity:
\begin{enumerate}
	\item Tableau composition computes
	$F=C^\dag EC$ in $\mathcal O(n^3)$ time
	\cite[\S III]{Aaronson2004}
	\cite[\S 2.3.2]{Gidney2021a},
	after which the columns $F^\dag\Z_iF$ are read from the tableau.
	The same preprocessing row-reduces $R$, builds the image-membership test
	and precomputes the phase data in
	\cref{eq:chi_phase_formula}; checking $\chi$ on a basis of $\ker R$
	is also $\mathcal O(n^3)$ time.
	\item For each nonzero coefficient $\alpha_L$:
	one need not explicitly form $\v u_L\in\mathbb{F}_2^r$.
	After row-reduction, the membership test and
	$\chi(\v u_L)$ can be precomposed with the binary symplectic
	vector of $L$. This gives linear and quadratic forms in $2k$ variables,
	respectively, which can be evaluated in $\mathcal O(k^2)$ time.
	\qedhere
\end{enumerate}
\end{proof}

\noindent
The following kernel fact is used in the proof.

\begin{lem}
\label{lem:kernel_chi_linear}
The restriction $\chi|_{\ker R}$ is linear
i.e.\ $\chi(\v k+\v l)
=\chi(\v k)+\chi(\v l)\pmod 2$ for $\v k, \v l \in \ker R$.
\end{lem}
\begin{proof}
$
(-1)^{\chi(\v k+\v l)}\I
\overset{\eqref{eq:chi_defn}}=
\Z_{\v k+\v l}'\Z_{\v v(\v k+\v l)}
=
\Z_{\v k}'\Z_{\v v(\v k)}
\Z_{\v l}'\Z_{\v v(\v l)}
\overset{\eqref{eq:chi_defn}}=
(-1)^{\chi(\v k)+\chi(\v l)}
\I$.
\end{proof}

\begin{lem}
\label{lem:trivial_syndrome_phase_evaluation}
Let
\[
F^\dag\Z_iF=\i^{\phi_i}\X_{\v x_i}\Z_{\v z_i},
\qquad
\v \phi=(\phi_1,\ldots,\phi_r)^\T\in\mathbb Z_4^r,
\qquad
X=[\v x_1,\ldots,\v x_r],
\qquad
Z=[\v z_1,\ldots,\v z_r].
\]
Then the function $\chi$ in \cref{eq:chi_defn} is
\begin{equation}
\chi(\v u)
:=
\frac{\v \phi^\T\v u-(X\v u)^\T(Z\v u)}2
+\sum_{i<j} u_i u_j\,\v z_i^\T\v x_j
\pmod 2,
\label{eq:chi_phase_formula}
\end{equation}
where the matrix--vector products $X\v u$ and $Z\v u$ are over $\mathbb{F}_2$,
then interpreted as $0/1$ integer vectors in the numerator.
The half is well-defined on this domain because
the numerator is always even.
\end{lem}
\begin{proof}
Consider first
\begin{align}
\Z_{\v u}'
&=\prod_{j \in \supp \v u} F^\dag \Z_j F \notag \\
&=\i^{\v \phi^\T \v u} \prod_{j \in \supp \v u}
\X_{\v x_j} \Z_{\v z_j}
\end{align}
We want to express this product
as one pure-X Pauli followed by one pure-Z Pauli.
In moving
all X components to the left and
all Z components to the right,
we accumulate factors of $-1$
because Paulis either commute or anticommute.
Specifically,
each time we move an $\X_{\v x_j}$
past a $\Z_{\v z_i}$ for $i <j$,
we pick up a $(-1)^{\v z_i^\T \v x_j}$ factor.
Thus,
the total number of $-1$ factors accumulated is
$\sum_{i <j} u_i u_j \v z_i^\T \v x_j$.
The X-component vector is $\sum_{j \in \supp \v u} \v x_j =X\v u$;
similarly for Z.
So the result is
\begin{equation}
\label{eq:raw_transformed_unencoded_product}
\Z'_{\v u}
=
\i^{\v \phi^\T\v u}
(-1)^{\sum_{i<j} u_i u_j\,\v z_i^\T\v x_j}
\X_{X\v u}\Z_{Z\v u}.
\end{equation}
Substituting this into \cref{eq:chi_defn},
we see that $\X_{X\v u}\Z_{Z\v u}$
equals $\Z_{\v v(\v u)}(\I^{\otimes r}\otimes L)$,
but without the powers of $\i$
from any $\Y$ tensor factors that appear in
$\Z_{\v v(\v u)}(\I^{\otimes r}\otimes L)$.
There are $(X\v u)^\T(Z\v u)$ such tensor factors
i.e.\ $\i^{(X\v u)^\T(Z\v u)} \X_{X\v u}\Z_{Z\v u}
=\Z_{\v v(\v u)}(\I^{\otimes r}\otimes L)$.
Substituting this into \cref{eq:raw_transformed_unencoded_product}
then into \cref{eq:chi_defn} implies the result.
\end{proof}
\section{Alternative Clifford-Error Analysis}
\label{sec:alternative_clifford_error_analysis}
In this section we go through
\cref{eg:d3_color_code_acceptance_probability},
but instead treat the Clifford error $E$
as a superposition of Paulis.
This alternative analysis is computationally slower but
recovers the same results,
is potentially more intuitive,
and can handle general unitary errors beyond Clifford.

Restricting focus to data qubits $(0, 1, 3)$,
write $E =-(\H_+)_0 (\H_-)_{1 3}$ as $(\Z\I\I) \H_-^{\otimes 3}$
up to a global phase, then
\begin{align}
	\H_-^{\otimes 3}
	&=\frac{1}{2^{3/2}} (\X -\Y)^{\otimes 3} \notag \\
	&=\frac{1}{2^{3/2}} `\big[
		(\X\X\X -\Y\Y\Y) +
		(\Y\Y\X -\X\X\Y) +
		(\Y\X\Y -\X\Y\X) +
		(\X\Y\Y -\Y\X\X)], \notag \\
	(\Z\I\I) \H_-^{\otimes 3}
	&=\frac{1}{2}(\Z\I\I -\I\Z\I -\I\I\Z -\Z\Z\Z)
	`\Big[\frac{1}{\sqrt{2}}(\X\X\X -\Y\Y\Y)].
\end{align}
The term in the square brackets
is a logical representative of $\H_+$,
so stabilises $\ket{\overline{T}}$;
we can therefore ignore it.
The remaining factor is
an equal superposition of four terms,
only the last of which is a valid logical operator (logical Z).
At this point we cannot yet conclude
$E$ leads to a logical error 1/4 of the time
because the four terms are still in superposition
[for example:
$\frac{1}{\sqrt{2}}(\X\ket{\T} +\Y\ket{\T})$ does not mean
an X or Y error on $\ket{\T}$ each with probability 1/2;
it equals $\H_+\ket{\T} =\ket{\T}$].
By \cref{lem:insert_projector},
we can consider the effect of this superposition on
an imaginary stabiliser measurement round.
The superposition collapses to one of the four terms
with equal probability:
the last term has trivial syndrome and leads to a $\overline{Z}$ error;
the other three terms have distinct nontrivial syndromes and lead to
abortion (in the case of error detection)
or identity (in the case of error correction).
Z errors maximally affect $\ket{T}$ states
i.e.\ rotate them to the orthogonal state.

To conclude,
the probability to observe a trivial syndrome is
$\pr(\v s =\v 0) =1/4$,
in which case the fidelity of
the resulting state to $\ket{\overline{T}}$ is 0.
For physical qubit count $n$,
this analysis runs in $\mathcal O(2^n)$ time
since $w \le n$ single-qubit Clifford errors
expands to a superposition of $\mathcal O(2^w)$ Pauli terms,
each of which must be analysed individually.
\section{Noise Model}
\label{sec:noise_model}
In this paper we use the standard depolarising noise model,
in which the following \emph{error processes} occur each with probability $p$
(the characteristic noise level of the circuit):
\begin{enumerate}
	\item each qubit preparation prepares the orthogonal state;
	\item each qubit measurement reports the opposite outcome;
	\item each 1-qubit gate (including idle) is followed by
	a Pauli drawn randomly from $`{\X, \Y, \Z}$;
	\item each 2-qubit gate is followed by a Pauli drawn randomly from
	$`{\I, \X, \Y, \Z}^{\otimes 2} \setminus `{\I^{\otimes 2}}$.
	\label[process]{item:2_qubit_depolarisation}
\end{enumerate}
These error processes give rise to 1, 1, 3, and 15 disjoint error events respectively
e.g.\ the 2-qubit depolarising channel (\cref{item:2_qubit_depolarisation}) after a CX
gives rise to the $\I_\d \Y_3$ error event that occurs with probability $p/15$
in \cref{fig:2_fault_configuration}.

Constructing a fault set $\mathcal F$ requires
the error events be \emph{independent} rather than \emph{disjoint}~\cite[\S 4.1]{Umbrarescu2026}.
Fortunately,
Granet et al.~\cite[\S III.C.1]{Granet2025} show how to factorise
any channel $\mathcal N(\rho) =\sum_{P \in \mathcal{P}_k} p_P P \rho P$
with disjoint Pauli error events into a sequential composition
$\bigcirc_{P \in \mathcal{P}_k \setminus `{\I^{\otimes k}}} \mathcal N_{q_P, P}$
of independent channels
$\mathcal N_{q_P, P}(\rho) :=(1 -q_P)\rho +q_P P \rho P$
each with one Pauli error event.
Applying this to the 1-qubit depolarising channel,
where $p_\I =1 -p$ and $p_\X =p_\Y =p_\Z =p/3$, yields
\begin{equation}
	q_\X =q_\Y =q_\Z =\frac{1}{2}`\Big[1 - `\Big(1 -\frac{4}{3}p)^{1/2}].
\end{equation}
Applying this to the 2-qubit depolarising channel,
where $p_P =p/15~\forall P \in \mathcal P_2 \setminus `{\I^{\otimes 2}}$,
yields
\begin{equation}
	q_P =\frac{1}{2}`\Big[1 -`\Big(1 -\frac{16}{15}p)^{1/8}]
	~\forall P \in \mathcal P_2 \setminus `{\I^{\otimes 2}}.
\end{equation}
We use these probabilities to make our disjoint error events independent.
\section{Logical Error Rate per Kept Shot}
\label{sec:logical_error_rate_per_kept_shot}
Consider all error events within
the final double check minus the second layer of T gates
e.g.\ for distance 3, this is the circuit before the dashed line in
\cref{fig:all_2_fault_configurations}.
Define the signature $\v \sigma$ of each fault over all detectors
that are not probabilistic for any considered error event;
namely:
	all detectors within,
	plus those corresponding to Z stabilisers after,
the double check.
The logical error rate per kept shot of that double check is
\begin{align}
\LER(p)
&=\pr(\text{malignant} \mid \text{kept}; p) \notag \\
&=\frac{\pr(\text{malignant} \cap \text{kept}; p)}{\pr(\text{kept}; p)}.
\end{align}
Recall that a shot is kept only if no detectors flip,
so we need only consider the set $\mathcal{U}
:=`{\mathcal{C} \subseteq \mathcal{F}: \v\sigma(\mathcal{C}) =\v 0}$
of all \emph{undetectable} configurations.
This turns out to be one of the most impactful optimisations in our analysis,
since $|\mathcal{U}|$ is much smaller than the total number $2^{|\mathcal F|}$ of configurations.

Write the second layer of T gates as $\T_{\mathcal L} \T^\dag_{\mathcal S}$
using $\mathcal L$ and $\mathcal S$ from \cref{lem:logical_H_is_transversal}.
Each resultant effect $P \in \mathcal P_n$
propagates past this layer into a Clifford error $E =\T_{\mathcal L} \T^\dag_{\mathcal S}
P (\T_{\mathcal L} \T^\dag_{\mathcal S})^\dag$.
It then has
probability $\kappa(P) :=`<\overline{T}| E^\dag \upPi_{\mathcal G}E |\overline{T}>$
of being kept (by \cref{lem:insert_projector}),
and probability $\mu(P) =\kappa(P)[1 -\varphi(E)]$ of being kept \emph{and} malignant,
where $\varphi(E) \in `{0, 1}$ is the fidelity referred to in \cref{prop:resulting_fidelity}.
E.g.\ the resultant effect $P =\X_{01}\Y_3$ in
\cref{eg:d3_color_code_acceptance_probability}
has $\kappa(P) =\mu(P) =1/4$.
Then
\begin{equation}
	\LER(p)
	=\frac
	{\sum_{\mathcal{C} \in \mathcal{U}} \pi_\mathcal{C}(p) \mu`\big(P(\mathcal{C}))}
	{\sum_{\mathcal{C} \in \mathcal{U}} \pi_\mathcal{C}(p) \kappa`\big(P(\mathcal{C}))}.
	\label{eq:ler_per_kept_shot}
\end{equation}
The sum $\sum_{\mathcal{C} \in \mathcal{U}}$ can contain up to $2^{|\mathcal F|}$ terms,
making direct evaluation intractable,
but since $\pi_\mathcal{C}(p) =\upTheta(p^{|\mathcal{C}|})$ for $p \ll 1$
we can approximate it by considering only low-fault-count configurations.
So partition $\mathcal{U} =\mathcal{U}_0 \cup \mathcal{U}_1 \cup \dots \cup \mathcal{U}_{|\mathcal F|}$
where $\mathcal{U}_k :=`{\mathcal{C} \in \mathcal{U}: |\mathcal{C}| =k}$,
and substitute this and \cref{eq:configuration_probability}
into \cref{eq:ler_per_kept_shot}:
\begin{align}
	\LER(p)
	&=\frac
	{\sum_{k=0}^{|\mathcal{F}|} g_k[\mu; p]}
	{\sum_{k=0}^{|\mathcal{F}|} g_k[\kappa; p]}, \\
	g_k[h; p]
	&:=\sum_{\mathcal{C} \in \mathcal{U}_k} h`\big(P(\mathcal{C}))
	\prod_{f \in \mathcal{C}} \frac{\pi_f(p)}{1 -\pi_f(p)}
	\qquad\forall h:\mathcal P_n \to [0, 1].
\end{align}
This is an exact expression;
its tractable approximation to order $l$ is given by
\begin{equation}
	\LER(p) =\frac
	{\sum_{k=0}^l g_k[\mu; p]}
	{\sum_{k=0}^l g_k[\kappa; p]} +\mathcal{O}(p^{l +1}).
\end{equation}

\end{document}